\documentclass[sigconf]{acmart}

\copyrightyear{2022}
\acmYear{2022}
\setcopyright{acmlicensed}
\acmConference[SIGMOD '22] {Proceedings of the 2022 International Conference on Management of Data}{June 12--17, 2022}{Philadelphia, PA, USA.}
\acmBooktitle{Proceedings of the 2022 International Conference on Management of Data (SIGMOD '22), June 12--17, 2022, Philadelphia, PA, USA}
\acmPrice{15.00}
\acmISBN{978-1-4503-9249-5/22/06}
\acmDOI{10.1145/XXXXXX.XXXXXX}

\usepackage{pgfplots}
\usepackage{pgfplotstable}
\pgfplotsset{compat=1.8}
\usepackage{subcaption}
\usepackage{xspace} 
\usepackage{xcolor}
\usepackage{algorithm}
\usepackage{algorithmicx}
\usepackage{float}
\usepackage{siunitx}
\usepackage{graphicx}
\usepackage{thm-restate}
\usepackage{blkarray}
\usepackage{balance}  
\usepackage{booktabs} 

\usepackage{marginnote} 
\usepackage{url}
\usepackage{enumitem}
\graphicspath{{./fig/}}
\usepackage{mathtools}
\usepackage[noend]{algpseudocode}
\newtheorem{theorem}{Theorem}
\newtheorem{lemma}{Lemma}
\newtheorem{corollary}{Corollary}
\newtheorem{definition}{Definition}

\newtheorem{observation}{Observation}
\usepackage{comment}
\usepackage{color}
\usepackage{cleveref}
\usepackage{tikz}
\usepackage{todonotes} 
\usetikzlibrary{fadings}
\usepackage[frozencache=true,cachedir=.]{minted}

\usepackage{dirtytalk}

\newcommand{\punt}[1]{}

\newtheorem{problem}{Problem}

\makeatletter
\def\@copyrightspace{\relax}
\makeatother

\newcommand{\defn}[1]       {{\textit{\textbf{\boldmath #1}}}}

\newcommand{\polylog}{\mathrm{polylog}}
\newcommand{\sort}{\mathrm{sort}}

\newcommand{\stream}{S}

\newcommand{\graph}{\mathcal{G}}
\newcommand{\nodes}{\mathcal{V}}
\newcommand{\edges}{\mathcal{E}}
\newcommand{\nodesize}{V}
\newcommand{\edgesize}{E}
\newcommand{\graphstream}{S}
\newcommand{\streamlength}{N}
\newcommand{\sketch}{\mathcal{S}}
\newcommand{\blocksize}{B}
\newcommand{\memsize}{M}
\newcommand{\disksize}{D}
\newcommand{\streamelement}{s}
\newcommand{\indexsubset}{\mathcal{B}}
\newcommand{\support}[1]{supp(#1)}
\newcommand{\prob}[1]{ \Pr \left [ #1 \right ]}
\newcommand{\veclength}{n}
\newcommand{\nodegroup}{\mathcal{U}}
\newcommand{\charvec}{f}
\newcommand{\bin}{bin}
\newcommand{\nodesizebound}{U}

\newcommand{\algname}{\textsc{StreamingCC}\xspace}
\newcommand{\sysname}{\textsc{GraphZeppelin}\xspace}
\newcommand{\modelname}{hybrid streaming }
\newcommand{\sketchname}{\textsc{CubeSketch}\xspace}
\newcommand{\sketchnames}{\textsc{CubeSketches}\xspace}
\newcommand{\oldsketchname}{$\ell_0$ sketch\xspace}
\newcommand{\batch}{batch\xspace}
\newcommand{\treename}{gutter tree\xspace}

\renewcommand{\subparagraph}[1]{\smallskip
\noindent
\emph{#1 }}

\newcommand{\ie}{\textit{i.e.,}~}

\newcommand{\edgeupdate}[1]{\textsc{edge\_update}(#1)\xspace}
\newcommand{\bufferinsert}[1]{\textsc{buffer\_insert}(#1)\xspace}
\newcommand{\dobatchupdate}{\textsc{do\_batch\_update}()\xspace}
\newcommand{\getbatch}{\textsc{get\_batch}()\xspace}
\newcommand{\sketchbatch}[1]{\textsc{update\_sketch\_batch}(#1)\xspace}
\newcommand{\listspanningforest}[1]{\textsc{list\_spanning\_forest}(#1)\xspace}
\newcommand{\cleanup}{\textsc{cleanup}()\xspace}

\newcommand{\etal}{\text{et al}.\xspace}

\date{}

\newcommand{\newtext}[1]{\textcolor{black}{#1}}

\newcommand{\namedcomment}[3]{{\sf \color{#2} #1: #3}}

\newcommand{\mab}[1]{\namedcomment{mab}{red}{#1}}
\newcommand{\mfc}[1]{\namedcomment{mfc}{purple}{#1}}
\newcommand{\david}[1]{\namedcomment{david}{red}{#1}}
\newcommand{\ahmed}[1]{\namedcomment{ahmed}{blue}{#1}}
\newcommand{\victor}[1]{\namedcomment{victor}{green}{#1}}
\newcommand{\evan}[1]{\namedcomment{evan}{orange}{#1}}
\newcommand{\abi}[1]{\namedcomment{abi}{cyan}{#1}}

\renewcommand{\mab}[1]{\todo[size=\tiny,color=green!40]{MAB: #1}}
\renewcommand{\mfc}[1]{\todo[size=\tiny,color=green!40]{MFC: #1}}
\renewcommand{\david}[1]{\todo[size=\tiny]{David: #1}}
\renewcommand{\ahmed}[1]{\todo[size=\tiny,color=yellow]{Ahmed: #1}}
\renewcommand{\victor}[1]{\todo[size=\tiny,color=yellow]{Victor: #1}}
\renewcommand{\evan}[1]{\todo[size=\tiny,color=red!40]{Evan: #1}}
\renewcommand{\abi}[1]{\todo[size=\tiny,color=red!40]{Abi: #1}}
\newcommand{\fixme}[1]{\todo[size=\tiny]{#1}}

\newcommand{\inline}[1]{\todo[inline,color=yellow,size=\tiny]{#1}}

\renewcommand{\epsilon}{\varepsilon}

\newcommand{\secref}[1]         {Section~\ref{sec:#1}}
\newcommand{\seclabel}[1]    {\label{sec:#1}}

\newcommand{\figref}[1]         {Figure~\ref{fig:#1}}

\newcommand{\tabref}[1]         {Table~\ref{tab:#1}}

\renewcommand{\eqref}[1]          {Eq.~\ref{eq:#1}}

\definecolor{bg}{rgb}{0.95,0.95,0.95}

\renewcommand{\paragraph}[1]{\vspace{0.1in}\noindent{\bf \boldmath #1}} 

\date{}

\begin{document}


\title{\sysname: Storage-Friendly 
Sketching
for Connected Components on Dynamic Graph Streams}



\author{David Tench}
\affiliation{%
  \institution{Rutgers University}
   \city{New Brunswick}
   \state{NJ}
   \country{USA}
}
\email{dtench@pm.me}

\author{Evan West}
\affiliation{%
  \institution{Stony Brook University}
  \city{Stony Brook}
  \state{NY}
  \country{USA}
}
\email{etwest@cs.stonybrook.edu}

\author{Victor Zhang}
\affiliation{%
  \institution{Rutgers University}
  \city{New Brunswick}
   \state{NJ}
   \country{USA}
}
\email{victor@vczhang.com}

\author{Michael A. Bender}
\affiliation{%
  \institution{Stony Brook University}
  \city{Stony Brook}
  \state{NY}
  \country{USA}
}
\email{bender@cs.stonybrook.edu}

\author{Abiyaz Chowdhury}
\affiliation{%
  \institution{Stony Brook University}
  \city{Stony Brook}
  \state{NY}
  \country{USA}
}
\email{abchowdhury@cs.stonybrook.edu}

\author{J. Ahmed Dellas}
\affiliation{%
  \institution{Rutgers University}
  \city{New Brunswick}
   \state{NJ}
   \country{USA}
}
\email{jad525@scarletmail.rutgers.edu}

\author{Martin Farach-Colton}
\affiliation{%
  \institution{Rutgers University}
  \city{New Brunswick}
   \state{NJ}
   \country{USA}
}
\email{martin@farach-colton.com}

\author{Tyler Seip}
\affiliation{
    \institution{MongoDB}
    \city{New York City}
    \state{NY}
    \country{USA}
}
\email{tylerjseip@gmail.com}

\author{Kenny Zhang}
\affiliation{%
  \institution{Stony Brook University}
  \city{Stony Brook}
  \state{NY}
  \country{USA}
}
\email{kzzhang@cs.stonybrook.edu}

\renewcommand{\shortauthors}{Tench, West, and Zhang et al.}


  \sloppy
\begin{abstract}

Finding the connected components of a graph is a fundamental
problem with uses throughout computer science and engineering.  
The task of computing connected components becomes more difficult when graphs are very large, or when they are dynamic, meaning the edge set changes over time subject to a stream of edge insertions and deletions.   A natural approach to computing the connected components on a large, dynamic graph stream is to buy enough RAM to store the entire graph.  However, the requirement that the graph fit in RAM is prohibitive for very large graphs.  Thus, there is an unmet need for systems that can process dense dynamic graphs, especially when those graphs are larger than available RAM.

We present a new high-performance streaming graph-processing system for computing the connected components of a graph.  This system, which we call \sysname, uses new linear sketching data structures (\sketchname) to solve the streaming connected components problem and as a result requires space asymptotically smaller than the space required for a lossless representation of the graph.  \sysname is optimized for massive dense graphs: \sysname can process millions of edge updates (both insertions and deletions) per second, even when the underlying graph is far too large to fit in available RAM.  As a result \sysname vastly increases the scale of graphs that can be processed.

\end{abstract}



\maketitle
\sloppy


\section{Introduction}
\seclabel{intro}

Finding the connected components of a graph is a fundamental 
problem with uses throughout computer science and engineering.  A recent survey by Sahu \etal~\cite{ubiquitous} of industrial uses of algorithms reports that, for both practitioners and academic researchers, connected components was the most frequently performed computation from a list of 13 fundamental graph problems that includes shortest paths, triangle counting, and minimum spanning trees. It has applications in scientific computing~\cite{scientific_computing1,scientific_computing2}, flow simulation~\cite{bioinformatics1}, metagenome assembly~\cite{metagenomics1,metagenomics2}, identifying protein families~\cite{protein1,bioinformatics2}, analyzing cell networks~\cite{cell}, pattern recognition~\cite{pattern1,pattern2}, graph partitioning~\cite{partition1,partition2}, random walks~\cite{randomwalks}, social network community detection~\cite{lee2014social}, graph compression~\cite{compression1,compression2}, medical imaging~\cite{tumor}, and object recognition~\cite{object}.  It is a starting point for strictly harder problems such as edge/vertex connectivity, shortest paths, and $k$-cores.  It is used as a subroutine for pathfinding algorithms such as Djikstra and $A^*$, some minimum spanning tree algorithms, and for many approaches to 
clustering~\cite{clustering1,clustering2,clustering3,clustering4,clustering5,clustering6}.

The task of computing connected components becomes more difficult when graphs are very large, or when they are \defn{dynamic}, meaning the edge set changes over time subject to a stream of edge insertions and deletions.  Applications on large graphs include metagenome assembly tasks that may include hundreds of millions of genes with complex relations~\cite{metagenomics1}, and large-scale clustering, which is a common machine learning challenge~\cite{clustering3}. Applications using dynamic graphs include identifying objects from a video feed rather than a static image~\cite{movingobject}, or tracking communities in social networks that change as users add or delete friends~\cite{dynamic_social,dynamic_web}.  And of course graphs can be both large and dynamic.  Indeed, Sahu \etal's~\cite{ubiquitous} survey reports that a majority of industry respondents work with large graphs ($> 1$ million nodes or $> 1$ billion edges) and a majority work with graphs that change over time.

A natural approach to computing the connected components on a large, dynamic graph stream is to buy enough RAM to store the entire graph. Indeed, dynamic graph stream processing systems such as Aspen and Terrace~\cite{aspen,terrace} can efficently query the connected components of a large graph subject to a stream of edge insertions and deletions when the graph fits in RAM.  However, the requirement that the graph fit in RAM is prohibitive for most large graphs:  for example, a graph with ten million nodes and an average degree of 1 million, using 2B to encode an edge, would require 10TB of memory.  We show in Section \ref{sec:experiments} that the Aspen and Terrace graph representations are significantly larger than this lower bound.

%

\begin{figure}[h]
  \centering
  \includegraphics[width=0.5\textwidth]{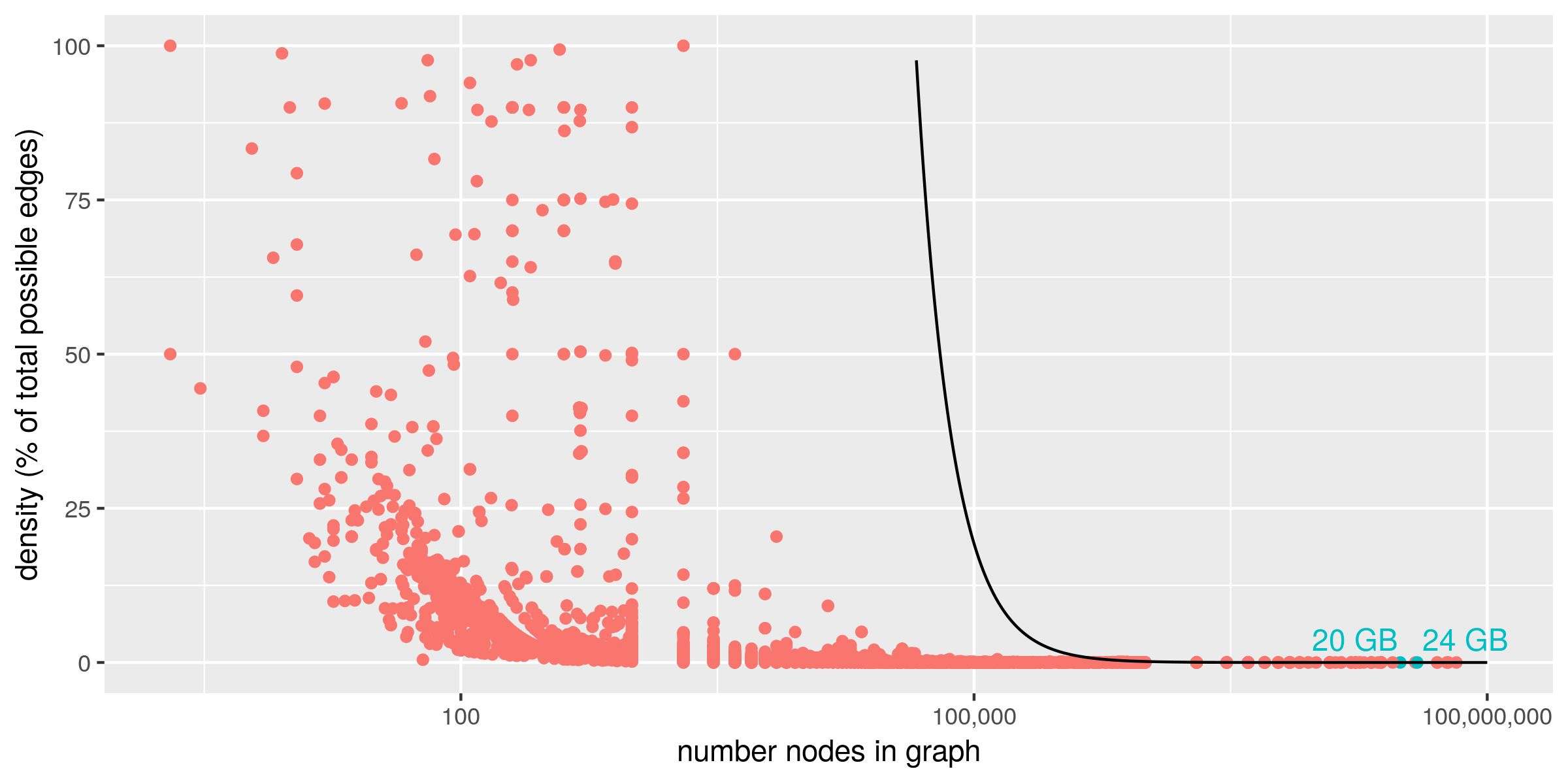}
\setlength{\belowcaptionskip}{-8pt}     
\caption{\textbf{Published graphs have few nodes or are sparse.}  Each point represents a graph data set from NetworkRepository. Any point below the dark line indicates a graph that can be represented as an adjacency list in 16GB of RAM.}\label{fig:graphs}
\end{figure}

In public graph-data-set repositories, most graphs are smaller than typical single-machine RAM sizes.
As \figref{graphs} illustrates, nearly all graphs in Network Repository~\cite{netrepo} can be stored as an adjacency list in less than 16GB.  This fixed memory budget furthermore implies that graphs with large numbers of vertices must be sparse.  
Similarly, the Stanford SNAP graph repository and the SuiteSparse repository have few graphs larger than 16GB, and graphs with many nodes are always extremely sparse.


Large, dense graphs, we argue, are absent from graph repositories not because they are unworthy of study, but because there are few tools to analyze them.  To illustrate: dense graphs do appear in Network Repository~\cite{netrepo}, but these graphs are never larger than a few GB; moreover, as the graphs' vertex count increases, the maximum density decreases such that the densest graphs never require more than 10 GB.  
A compelling explanation for the absence of large, dense graphs is selection bias: interesting dense graphs exist at all scales, but large, dense graphs are discarded as computationally infeasible and consequently are rarely published or analyzed. \newtext{Moreover, some large dense graphs are known to exist as proprietary datasets: for instance, Facebook works with graphs with $40$ million nodes and $360$ billion edges. These graphs are processed at great cost on large high-performance clusters, and are consequently not released for general study.\cite{facebookdense}}


Thus, there is an unmet need for systems that can process dense graphs, especially when those graphs are larger than available RAM.  Existing systems are not designed for large, dense, dynamic graph streams and instead optimize for other use cases.  Aspen and Terrace are optimized for large, sparse, dynamic graphs that completely fit in RAM, but their performance degrades significantly on dense graphs and graphs larger than RAM.  There is a deep literature on parallel systems for connected components computation in multicore~\cite{greiner1994comparison}, GPU~\cite{gpucc}, and distributed settings~\cite{krishnamurthy1997connected,buvs2001parallel} but these focus on static graphs which fit in RAM.  Many external memory~\cite{brodal2021experimental} and semi-external memory~\cite{abello2002functional} systems focus on graphs that are too large for RAM and must be stored on disk, but none of these systems focus on graphs whose edges can be deleted dynamically.


In this paper, we explore the general problem of connected components on large, dense, dynamic graphs.  We introduce \sysname, which computes the connected components of graph streams using a $O(\nodesize/\log^3(\nodesize))$-factor less space than an explicit representation of the graph.  \sysname uses a new $\ell_0$-sketching data structure that outperforms the state of the art on graph sketching workloads.  Additionally, 
\sysname employs node-based buffering strategies that improve I/O efficiency.  
These techniques allow \sysname to scale better than existing systems in several settings.  First, for in-RAM computation, \sysname's small size means it can process larger, denser graphs than Aspen or Terrace: specifically, dense graphs twice as large as Aspen and at least 40 times larger than Terrace given 64 GB of RAM.  Moreover, even if the input graph fits in RAM on all systems, \sysname is up to 2.5 times faster than Aspen and 36 times faster than Terrace on large dense graphs.  Finally, \sysname scales to SSD at the cost of a 29\% decrease to ingestion rate, and is more than two orders of magnitude faster than Aspen and Terrace, which suffer significant performance degradation when scaling out of RAM.  


\sysname employs a new sketch algorithm, overcoming a computational bottleneck of existing linear sketching techniques in the semi-streaming graph algorithms literature~\cite{l0sketch}.  The asymptotically best existing streaming connected components algorithm is Ahn \etal's \algname~\cite{Ahn2012, nelson2019optimal}, which has asymptotically low space and update time complexity.  \algname relies on $\ell_0$-sampling, which it uses to sample edges across arbitrary graph cuts.  However, the best known $\ell_0$-sampling algorithm suffers from high constant and polylogarithmic factors in its space and update time, as we show in Section~\ref{sec:l0revis}.  This overhead makes any implementation of the \algname data structure infeasibly slow and large. \sysname employs what we call \sketchname, a specialized $\ell_0$-sampling algorithm for sampling edges across graph cuts, to solve the connected components problem.  For large graphs \sketchname uses 4 times less space than the best general $\ell_0$-sampling algorithm and can process updates more than three orders of magnitude faster.


\sysname also uses new write-optimized data structures to overcome prohibitive resource requirements of existing semi-streaming algorithms.  Streaming algorithms have had a significant impact in large part because they require a small (polylogarithmic) amount of RAM.  In contrast, graph semi-streaming algorithms have higher RAM requirements: for most problems on a graph with $\nodesize$ nodes, sublinear RAM is insufficient to even represent a solution so $O(\nodesize \polylog(\nodesize))$ RAM is typically assumed.  
With the large $\polylog$ factors, this is often more RAM than is feasible in practice; see \secref{preliminaries}.
We propose the \modelname model, which enjoys the memory advantage of the streaming model while allowing enough space in external memory to compute on dynamic graph streams. In this model there is still $O(\nodesize \polylog(\nodesize))$ space available, but only $O(\polylog(\nodesize))$ of this space is RAM and the rest is disk, which may only be accessed in $O(\polylog(\nodesize))$-size blocks.  The simultaneous challenges in this model are to design algorithms that use small total space but also have low I/O complexity.  While existing graph semi-streaming algorithms use small space, their heavy reliance on hashing and random access patterns make them slow on disk.  We show that \sysname is simultaneously a space-optimal in-RAM semi-streaming algorithm and an I/O efficient external memory algorithm for the connected components problem. 
We also validate its performance experimentally, showing that \sysname can operate on modern consumer solid-state disk, increasing the scale of dynamic graph streams that it can process while incurring only a $29\%$ cost to stream ingestion rate.

\paragraph{Results.}
In this paper we establish the following:
\begin{itemize}[leftmargin=*,topsep=.5em]
 \setlength\itemsep{.5em}

    \item \textbf{\sysname: } We present a new high-performance streaming graph-processing system for computing the connected components of a graph.  This system, which we call \sysname, uses new linear sketching data structures (\sketchname, described below) to solve the streaming connected components problem and as a result requires a $O(\nodesize/\log^3(\nodesize))$-factor less space than any lossless representation of the graph.  \sysname is optimized for massive dense graphs: \sysname can process millions of edge updates (both insertions and deletions) per second, even when the underlying graph is far too large to fit in available RAM.  As a result \sysname vastly increases the scale of graphs that can be processed. 
    \vspace*{.4em}

    \item \textbf{\sketchname: $\ell_0$-sampling optimized for graph connectivity sketching.} We give a new $\ell_0$-sampling algorithm, \sketchname, for vectors of integers mod $2$.  Given a vector of length $\veclength$ and failure probability $\delta$, \sketchname uses $O(\log^2(\veclength)\log(1/\delta))$ bits of space and $O(\log(\veclength)\log(1/\delta))$ average time per update, which is a factor of $O(\log(\veclength))$ faster than the best existing $\ell_0$-sampler for general vectors~\cite{l0sketch}.
    \vspace*{.4em}

    \sketchname is a key subroutine in \sysname, where it is used to sample graph edges across arbitrary cuts as part of connected components computation.  Here it is used to sketch vectors of length ${{\nodesize}\choose {2}} = O(\nodesize^2)$, where $\nodesize$ denotes the number of nodes in the graph. We show experimentally that \sketchname is more than 3 orders of magnitude faster than the state-of-the-art $\ell_0$ sampling algorithm on graph streaming workloads.  
    
    
    In addition to the $O(\log(\nodesize))$-factor speedup, several non-asymptotic factors contribute to this performance improvement as well.  First, the existing algorithm's average update cost is dominated by $O(\log(\nodesize)\log(1/\delta))$ division operations, while \sketchname's average update cost is dominated by $O(\log(1/\delta))$ bitwise XOR operations, which are much faster.  In addition, the general algorithm performs 128-bit arithmetic operations (including division) when processing graphs with more than $10^5$ nodes, whereas \sketchname can use standard 64-bit operations to achieve the same error probability.  Finally, both algorithms match the asymptotic space lower bound but \sketchname uses roughly 4 times less space than the general algorithm.
    
    \item \textbf{Asymptotic guarantees of \sysname: space-optimality, I/O efficiency, $O(\log^3(V))$ time per update.}  \sysname's core algorithm matches the $O(\nodesize \log^3(\nodesize))$-bit space lower bound for the streaming connected components problem, and its average per-update time cost of $O(\log(\nodesize))$ is $O(\log(\nodesize)))$ times faster than the best existing algorithm~\cite{Ahn2012}.  Additionally, \sysname can efficiently ingest stream updates even when its sketch data structure is too large to fit in RAM: its I/O complexity is $sort(\text{length of stream}) + O(\nodesize/\blocksize \log^3(\nodesize) + \nodesize \log^*(\nodesize))$ and for realistic block sizes it is an I/O-optimal external-memory algorithm~\cite{ChiangGoGr95}.  As a result, given a fixed amount of RAM and disk, \sysname is capable of efficiently computing the connected components of larger graphs than existing algorithms in the streaming or external memory models.
    
    \item \textbf{Empirical achievements of \sysname: better scaling for in-memory, out-of-core, and parallel computation, and undetectable failure probability.}  \sysname's \sketchname-based design increases the size of input graphs that can be processed, scales well to persistent memory, and facilitates parallelism in stream ingestion.  As a result, \sysname can ingest 2-5 million edge updates per second on a single scientific workstation (see Section~\ref{sec:experiments}), both when its data structures reside completely in RAM and also when they reside on fast disk. As a result of these advantages, \sysname is faster and more scalable than the state of the art on large, dense graphs:


    \begin{itemize}[topsep=.5em]
        \item \textbf{\sysname handles larger graphs for in-RAM computation.}  \sysname's space-efficient \sketchname allows it to process graph streams larger than can be stored explicitly in a fixed amount of RAM and give it an asymptotic ${O}(\nodesize/ \log^3(\nodesize))$ space advantage over state-of-art systems on dense graphs.  Given the polylogarithmic factors and constants, we need to determine the actual crossover point where \sysname processes graphs more compactly than Aspen and Terrace. We show empirically that this crossover point occurs when the space budget is between 32 and 64 gigabytes. That is, for dense graphs on several hundred thousand nodes, \sysname is $40\%$ more compact than Aspen and several times more compact than Terrace, and this advantage only increases for larger space budgets or input sizes. Additionally, for dense graph streams on $2^{18}$ nodes \sysname ingests updates 24 times faster than Terrace and twice as fast as Aspen.
        \item \textbf{\sysname can use persistent memory to handle even larger graphs.}  \sysname's node-based work buffering strategy facilitates out-of-core computation, allowing \sysname to use SSD to increase the scale of graph streams it can process while incurring a small cost to performance.  We show experimentally that \sysname ingests updates more than two orders of magnitude faster than Aspen and Terrace when all systems swap to disk, and that using SSD slows \sysname stream ingestion by only 29\%.
        \item \textbf{\sysname's stream ingestion is highly parallel.}  \sysname employs a node-based work buffering strategy that facilitates parallelism and improves data locality.  We show experimentally that \sysname's multithreaded stream ingestion system scales well with more threads:  its ingestion rate is 25 times higher with 46 threads than an optimized single-thread implementation.  
        \item \textbf{\sysname's theoretical failure probability is undetectable in practice.}  \sysname and similar graph sketching approaches achieve their remarkable space efficiency at the cost of a random chance of failure.  We show empirically that \sysname's observed failure rate is even lower than the proved (polynomially small) upper bound: in fact, for 5000 trials on real-world and synthetic graphs it never failed.
    \end{itemize}
\end{itemize}

\vspace*{-1em}

\section{Preliminaries}
\label{sec:preliminaries}

\subsection{Graph Streaming \& Hybrid Graph Streaming}
\label{subsec:models}

In the \defn{graph semi-streaming} model~\cite{semistreaming1,semistreaming2} (sometimes just called the \defn{graph streaming} model), an algorithm is presented with a \defn{stream} $\graphstream$ of updates (each an edge insertion or deletion) where the length of the stream is $\streamlength$. 
Stream $\graphstream$ defines an input graph $\graph = (\nodes,\edges)$ with $\nodesize = |\nodes|$ and $\edgesize = |\edges|$.  The challenge in this model is to compute (perhaps approximately) some property of $\graph$ given a single pass over $\graphstream$ and at most $O(\nodesize \polylog(\nodesize))$ words of memory.  
Each update has the form $((u,v), \Delta)$ where $u,v \in \edges, u \neq v$ and $\Delta \in \{-1,1\}$ where $1$ indicates an edge insertion and $-1$ indicates an edge deletion.  Let $\streamelement_i$ denote the $i$th element of $\graphstream$, and let $\graphstream_i$ denote the first $i$ elements of $\graphstream$.  
Let $\edges_i$ be the edge set defined by $\stream_i$, i.e., those edges which have been inserted and not subsequently deleted by step $i$.  The stream may only insert edge $e$ at time $i$ if $e \notin \edges_{i-1}$, and may only delete edge $e$ at time $i$ if $e \in \edges_{i-1}$.

In Section ~\ref{sec:io} we additionally use a new variant of the graph semi-streaming model, which we call the \defn{hybrid graph streaming setting} (since it incorporates some components of the external memory model~\cite{externalmemory} into the semi-streaming model).  In this setting, there is an additional constraint on the type of memory available for computation: only $\memsize = \Omega(\polylog(\nodesize))=o(\nodesize)$ RAM is available, and $\disksize = O(\nodesize \polylog(\nodesize))$ disk space is available. A word in RAM is accessed at unit cost, and disk is accessed in blocks of $\blocksize = o(\memsize)$ words at a cost of $\blocksize$ per access. Any semi-streaming algorithm can be run with this additional constraint, but may become much slower if the algorithm makes many random accesses to disk.  The algorithmic challenge in the hybrid graph streaming setting is to minimize time complexity (of ingesting stream updates and returning solutions) in addition to satisfying the typical limited-space requirement of the data stream setting. In Section ~\ref{sec:io} we show how \sysname can be adapted to this model, and is both a space-optimal single pass streaming algorithm with $O(\log^2(\nodesize))$ update time and an I/O efficient external memory algorithm. 

\begin{problem}[\textbf{The streaming Connected Components problem.}]
Given a insert/delete edge stream of length $\streamlength$ that defines a graph $\graph = (\nodes, \edges)$, return a insert-only edge stream that defines a spanning forest of $\graph$.
\end{problem}

\subsection{Prior Work in Streaming Connected Components}
\label{subsec:alg}
We summarize \algname, Ahn \etal's~\cite{Ahn2012} semi-streaming algorithm for computing a spanning forest (and therefore the connected components) of a graph.

For each node $v_i$ in $G$, define the \defn{characteristic vector} $\charvec_i$ of $v_i$ to be a 1-dimensional vector indexed by the set of possible edges in $\graph$.  $\charvec_i[(j,k)]$ is only nonzero when $i = j$ or $i = k$ and edge $(j,k) \in \edges$.  That is, $\charvec_i \in \{-1,0,1\}^{\nodesize \choose 2}$ s.t. for all $0 \leq j < k < {\nodesize \choose 2}$: 
$$\charvec_i[(j,k)] = \left\{ \begin{array}{ll}
            1 & \quad i = j \text{ and }(v_j, v_k) \in \edges \\
            -1 & \quad i = k \text{ and }(v_j, v_k) \in \edges \\
            0 & \quad \text{otherwise} 
            
        \end{array}\right\}
$$


Crucially, for any $S \subset \nodes$, the sum of the characteristic vectors of the nodes in $S$ is a direct encoding of the edges across the cut $(S, \nodes \setminus S$).  That is, let $x = \sum_{v \in S} \charvec_v$ and then 
$|x[(j,k)]| = 1$ iff $(j,k) \in E(S, \nodes \setminus S)$.

Using these vectors, we immediately have a (very inefficient) algorithm for computing the connected components from a stream:  Initialize $\charvec_i = \{0\}^{\nodesize \choose 2}$ for all $i$. For each stream update $s = ((u,v), \Delta)$, set $\charvec_u[u,v] += \Delta$ and $\charvec_v[u,v] += -\Delta$.

After the stream, run Boruvka's algorithm~\cite{boruvka} for finding a spanning forest as follows.  For the first round of the algorithm, from each $a_i$ arbitrarily choose one nonzero entry $(w,y)$ (an edge in $\edges$ s.t. w = i or y = i).  Add $e_i$ to the spanning forest.  For each connected component $C$ in the spanning forest, compute the characteristic vector of $C$: $a_C = \sum_{v\in C} \charvec_v$.  Proceed similarly for the remaining rounds of Boruvka's algorithm: in each round, choose one nonzero entry from the characteristic vector of each connected component and add the corresponding edges to the spanning forest.  Sum the characteristic vectors of the component nodes of the connected components in the spanning forest, and continue until no new merges are possible.  This will take at most $O(\log(\nodesize))$ rounds.

The key idea to make this a small-space algorithm is to use ``$\ell_0$-sampling''~\cite{l0sketch} to run this version of Boruvka's algorithm by compressing each characteristic vector $\charvec_i$ into a data structure of size $O(\log^2(\nodesize))$ that can return a nonzero entry of $\charvec_i$ with constant probability.

\begin{definition}
\label{def:l0}
A sketch algorithm is a $\delta$ $\ell_0$-sampler if it is
\begin{enumerate}
    \item \textbf{Sampleable:} it can take as its input a stream of updates to the coordinates of a non-zero vector $a$, and output a non-zero coordinate $(j, \charvec[j])$ of $\charvec$.  $\sketch(\charvec)$ denotes the sketch of vector $\charvec$.
    \item \textbf{Linear:} for any vectors $\charvec$ and $g$, $\sketch(\charvec) + \sketch(g) = \sketch(\charvec+g)$ and this operation preserves sampleability, \ie $\sketch(\charvec+g)$ can output a nonzero coordinate of vector $\charvec+g$.
    \item \textbf{Low Failure Probability.} the algorithm returns an incorrect or null answer with probability at most $\delta$.
\end{enumerate}
\end{definition}

For all $\ell_0$-samplers in this paper, $\sketch(\charvec)$ is a vector and adding two sketches is equivalent to adding their vectors elementwise.

\begin{lemma} (Adapted from~\cite{l0sketch}, Theorem 1):
Given a 2-wise independent hash family $\mathcal{F}$ and an input vector of length $n$, there is an $\delta$ $\ell_0$-sampler using $O(\log^2(n)\log(1/\delta))$ bits of space.
\end{lemma}

We denote a $\ell_0$ sketch of a vector $x$ as $\sketch(x)$.  Since the sketch is linear, $\sketch(x) + \sketch(y) = \sketch(x+y)$ for any vectors $x$ and $y$.  This allows us to process stream updates as follows: we maintain a running sum of the sketches of each stream update, which is equivalent to a sketch of the vector defined by the stream. That is, let $a_i^t$ denote $a_i$ after stream prefix $\stream_t$. For the $j$th stream update $\streamelement_j = ((i,x),\Delta)$ we obtain $\sketch(\charvec_i^j) = \sketch(\streamelement_j) + \sketch(\charvec_i^{j-1})$.  

Linearity also allows us to emulate the merging step of Boruvka's algorithm by summing the sketches of all nodes in each connected component.  We require $O(\log(\nodesize))$ independent $\ell_0$ sketches for each $v \in \nodes$, one for each round\footnote{In the original paper the authors note that adaptivity concerns require the use of new sketches for each round of Boruvka's algorithm.}, and each must succeed with probability $1-1/100$ so the size of the sketch data structure for each node is $O(\log^3(\nodesize))$.  We refer to the sketch data structure for each node as a \defn{node sketch} and each of its $O(\log(\nodesize))$ $\ell_0$-subsketches as \sketchnames.  The total size of the entire data structure is $O(\nodesize \log^3(\nodesize))$.  Recent work~\cite{nelson2019optimal} has shown that this is optimal.

\newtext{
The above description assumes that the exact number of nodes $\nodesize$ is known a priori. This is not strictly necessary: All we need is a loose upper bound on the number of nodes we will eventually see.  Given an upper bound $U$ s.t. $\nodesize \leq \nodesizebound \leq \nodesize^c$ for some constant $c$, we can simply define $f_i$ to have length ${\nodesizebound \choose 2}$.  The node sketch of $f_i$ then has size $O(\log^3(\nodesizebound^2)) = O(\log^3(\nodesize))$. 
We create a node sketch for $v_i$ the first time it appears in a stream update $(v_i, v_j)$ so the total space cost is still $O(\nodesize \log^3(\nodesize))$.
Similarly, even if nodes are identified in the input stream as arbitrary strings instead of integer IDs in the range $[\nodesize]$, we can use a hash function with range $[O(\nodesizebound^2)]$ to ensure that every node gets a unique integer ID with high probability.
}

\section{$\ell_0$-sampling Revisited}
\label{sec:l0revis}
Existing $\ell_0$-sampling algorithms are asymptotically small and fast to update, but in practice high constant and logarithmic overheads in size and update time prevent these algorithms from being useful for a streaming connected components algorithm.  We now review some details of the best known $\ell_0$-sampling algorithm and demonstrate experimentally that using it to emulate Boruvka's algorithm for graph connectivity would be prohibitively slow and would require an enormous amount of space.  Then we introduce an $\ell_0$-sketching algorithm which exploits the structure of the connected components problem to improve performance, and experimentally demonstrate that it is 4 times smaller and 3 orders of magnitude faster to update than the state of the art.

\newtext{
The best known $\ell_0$-sampling algorithm~\cite{l0sketch} is summarized in Figure~\ref{fig:agm_code}. Given a vector $\charvec \in \mathbb{Z}^\veclength$, the data structure consists of a matrix of $\log(\veclength)$ by $q\log(1/\delta)$ "buckets" (for some small constant $q$). Each bucket represents the values at a random subset of positions of $\charvec$. This representation is lossy: we can recover a nonzero element of $\charvec$ from bucket $\indexsubset_{i,j}$ only when a single position in $\indexsubset_{i,j}$ is nonzero. Equivalently, the support of $\indexsubset_{i,j}$, denoted by $\support{\indexsubset_{i,j}}$, is 1. If $\support{\indexsubset_{i,j}} = 1$, we say that $\indexsubset_{i,j}$ is \defn{good}, and say that it is \defn{bad} otherwise. With probability $1-\delta$, $\exists i,j$ s.t. $\indexsubset_{i,j}$ is good and therefore we can recover a nonzero value from $\charvec$. 
Each bucket includes a \defn{checksum} that indicates whether it is good with high probability.
}

\begin{figure}
    \centering
    \includegraphics[width=0.35\textwidth]{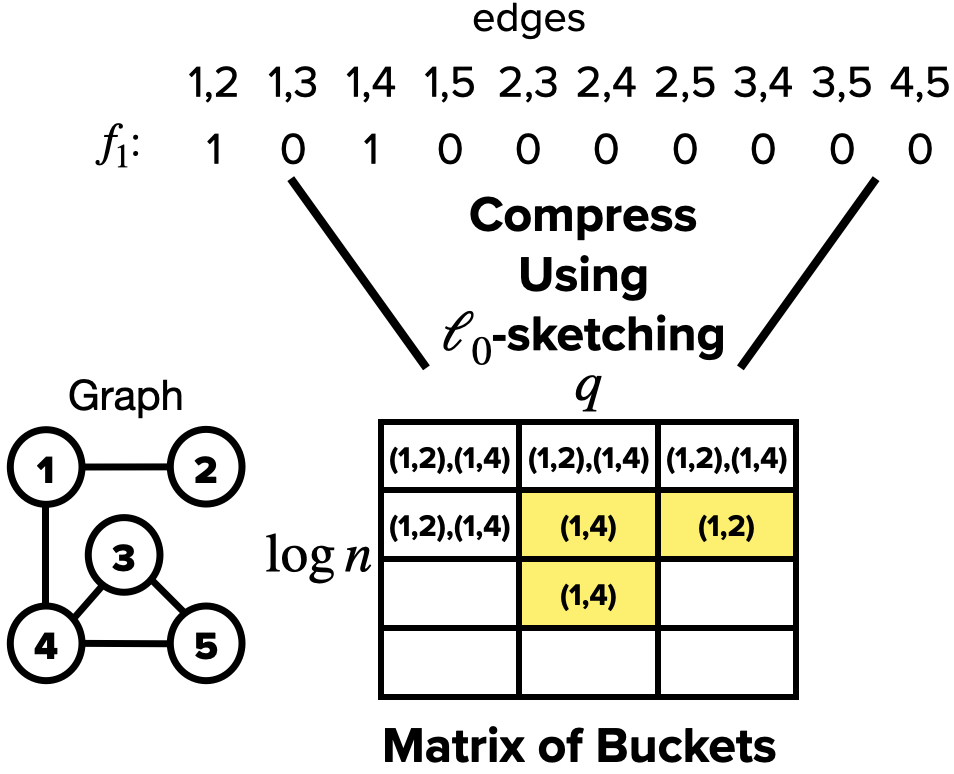}
    \setlength{\belowcaptionskip}{-16pt} 
    \caption{Compressing a characteristic vector. Each highlighted cell contains one nonzero element from the vector and can be sampled, yielding an edge incident to node 1.}
    \label{fig:sketch_fig}
\end{figure}

\newtext{
Each bucket $\indexsubset_{i,j}$ contains three values: $a_{i,j}, b_{i,j},$ and $c_{i,j}$. If $\indexsubset_{i,j}$ is good, then the checksum test on line \ref{line:oldchecksum} passes and $\charvec[b_{i,j}] = a_{i,j}/b_{i,j}$. If the checksum test fails $\indexsubset_{i,j}$ is bad.
}

\newtext{
When a stream update $(e, \Delta)$ arrives, its membership in each bucket is determined using the hash function on line \ref{line:oldhash}: if $hash(e)\equiv0\pmod{2^i}$ then $e$ is in $\indexsubset_{i,j}$. If it is in bucket $\indexsubset_{i,j}$, it is applied to $a_{i,j}, b_{i,j},$ and $c_{i,j}$ according to the logic on lines \ref{line:1}, \ref{line:2}, and \ref{line:3}. When the sketch is queried, it checks whether each bucket passes the checksum test on line \ref{line:oldchecksum}. If some bucket passes this test, its sampled value is returned. Figure \ref{fig:sketch_fig} gives an example of this process. For a more thorough analysis of this algorithm see~\cite{l0sketch}.
}


\begin{figure}[!t]
\begin{small}
\begin{algorithmic}[1]
\Function {update\_sketch}{idx, $\Delta$}
\Comment{Add $\Delta$ to vector index `idx'}
    \ForAll {col $\in [0, q\log(1/\delta))$}
        \State col\_hash $\gets$ hash(col, idx) \label{line:oldhash}
        \State row $\gets 0$
        \State checksum $\gets r[\text{col}]^\text{idx} \mod p$
        \While {row == 0 OR col\_hash[row-1] == 0}
            \State col[row].a $\gets \text{col[row].a} + \text{idx} \times \Delta$ \label{line:1}
            \State col[row].b $\gets \text{col[row].b} + \Delta$ \label{line:2}
            \State col[row].c $\gets \text{col[row].c} + \Delta \times \text{checksum}$ \label{line:3}
            \State row $\gets \text{row} + 1$
        \EndWhile
    \EndFor
\EndFunction

\Function{query\_sketch}{ }
\Comment{Get a non-zero vector index}
    \ForAll {col $\in [0, q\log(1/\delta))$}
        \ForAll {bucket $\in$ col}
            \State value $\gets \text{bucket.a} / \text{bucket.b}$
            \If {value is integer \hspace{0.25cm} AND \hspace{0.25cm} bucket.c == $\text{bucket.b} \times r[\text{col}]^\text{value}\mod p$} \label{line:oldchecksum}
                \State \textbf{return} \{value, bucket.b\}
                \Comment{Found a good bucket, done}
            \EndIf
        \EndFor
    \EndFor
    \State \textbf{return} sketch\_failure
    \Comment{All buckets bad}
\EndFunction
\end{algorithmic}
\end{small}
\setlength{\belowcaptionskip}{-16pt} 
\caption{State---of---the---art $\ell_0$-sampling algorithm.}
\label{fig:agm_code}
\end{figure}


\paragraph{Existing \boldmath $\ell_0$-samplers are slow to update for graph streaming workloads.}
\newtext{Note in line \ref{line:3} that updating $c_{i,j}$ of bucket $\indexsubset_{i,j}$ requires modular exponentiation, necessitating $O(\log(\veclength))$ multiplication operations and $O(\log(\veclength))$ modulo operations (where the modulus is a large prime). As a result, in the worst case this algorithm performs $O(\log(\veclength)\log(1/\delta))$ arithmetic operations per stream update.  In the average case, the update modifies only $O(\log(1/\delta))$ buckets, however, the cost to generate checksums is still $O(\log(\veclength)\log(1/\delta))$.  Moreover, for sufficiently large vectors, this modular exponentiation must be done on integers larger than a 64-bit machine word, drastically increasing computation time in practice.
}

\begin{figure}
\begin{center}

\begin{tabular}{ |c|c|c|c| } 
 \hline
 Vector Length & Standard $\ell_0$ & \sketchname & Speedup\\ 
 \hline
 $10^3$ & 223,000 & 7,322,000 & 33.1 x\\ 
 \hline
 $10^4$ & 124,000 & 5,180,000 & 42.3 x\\ 
 \hline
 $10^5$ & 54,800 & 4,384,000 & 79.8 x\\ 
 \hline
 $10^6$ & 29,300 & 3,730,000 & 127 x\\ 
 \hline
 $10^7$ & 20,900 & 3,177,000 & 151 x\\ 
 \hline
 $10^8$ & 16,300 & 2,825,000 & 172 x\\ 
 \hline
 $10^9$ & 13,100 & 2,587,000 & 196 x\\ 
 \hline
 $10^{10}$ & 1,350 & 2,272,000 & 1,680 x\\ 
 \hline
 $10^{11}$ & 918 & 2,108,000 & 2,300 x\\ 
 \hline
 $10^{12}$ & 833 & 1,963,000 & 2,350 x\\ 
 \hline
\end{tabular}
\end{center}
\setlength{\belowcaptionskip}{-16pt} 
\caption{\sketchname is faster than standard $\ell_0$ sketching. Ingestion rates (in updates/second) are listed for both $\ell_0$ sketching methods.
}
\label{fig:l0speedup}
\end{figure}

The ``Standard $\ell_0$'' column of Figure ~\ref{fig:l0speedup} displays the single-threaded ingestion rate in updates per second of the state-of-the-art $\ell_0$-sampling algorithm for vectors of various sizes. These results were obtained on a Dell Precision 7820 with 24-core 2-way hyperthreaded Intel(R) Xeon(R) Gold 5220R CPU @ 2.20GHz and 64GB 4x16GB DDR4 2933MHz RDIMM ECC Memory. Note how ingestion rate decreases as vector length increases, and in particular there is a catastrophic slowdown at vector length $10^{10}$.  This dramatic decrease in ingestion rate is due to the need to perform modular exponentiation on integers larger than $2^{64}$, requiring the use of 128-bit integers thus slowing computation.  When sketching characteristic vectors of length $O(\nodesize^2)$ for streaming connected components, 128-bit integers are required when $\nodesize \geq 10^5$.


When using $\ell_0$-sampling for Boruvka emulation, each stream update $((u,v), \Delta)$ must be applied to the node sketches of $u$ and $v$.  For any node $u$, the node sketch of $u$ is made up of $\log(\nodesize)$ $\ell_0$-sketches of $a_u$. Each of these $\ell_0$-sketches has a failure rate of $\delta = 1/100$ and, therefore, a width of $\log(1/\delta) = 7$. Processing a stream update requires $2 \cdot 7 \cdot O(\log^2(|a_u|) = 28 \cdot O(\log^2(\nodesize))$ multiplication and modulo operations.   For a graph with a million nodes, \algname must apply each update to $28$ sketch vectors of length $10^{12}$, so it can process roughly $800/28 = 29$ edge updates per second.

\begin{figure}
\begin{center}
\begin{tabular}{ |c|c|c|c| } 
\hline
Vector Length & Standard $\ell_0$ & \sketchname & Size Reduction\\ 
\hline
$10^{3}$ & 2.30KiB & 1.21KiB & 1.9 x\\
\hline
$10^4$ & 4.98KiB & 2.34KiB & 2.1 x\\
\hline
$10^{5}$ & 7.23KiB & 3.43KiB & 2.1 x\\
\hline
$10^6$ & 9.90KiB & 4.73KiB & 2.1 x\\
\hline
$10^{7}$ & 14.1KiB & 6.79KiB & 2.1 x\\
\hline
$10^8$ & 17.8KiB & 8.58KiB & 2.1 x\\
\hline
$10^{9}$ & 21.9KiB & 10.6KiB & 2.1 x\\
\hline
$10^{10}$ & 55.9KiB & 13.6KiB & 4.1 x\\
\hline
$10^{11}$ & 66.0KiB & 16.1KiB & 4.1 x\\
\hline
$10^{12}$ & 77.0KiB & 18.8KiB & 4.1 x\\
\hline
\end{tabular}
\end{center}

\setlength{\belowcaptionskip}{-8pt} 
\caption{\sketchname is significantly smaller than standard $\ell_0$ sketching.
Sizes are listed for both $\ell_0$ sketching methods.}
\label{fig:l0size}
\end{figure}


\paragraph{Existing \boldmath $\ell_0$-samplers are large for graph streaming workloads.}
Each node sketch consists of $\log(\nodesize)$ \oldsketchname{}es and each $\ell_0$-sketch is a vector of $7c\log(\nodesize^2) = 14c\log(\nodesize)$ buckets. Each bucket is composed of three integers so a node sketch consists of $42c \log^2(\nodesize)$ integers. 
As noted above, 128-bit(16B) integers are necessary when $\nodesize \geq 10^5$, so for $c = 2$ the size of a node sketch is $1344\log^2(\nodesize)$B.  
Since there is a node sketch for each node in the graph, the entire streaming data structure has size $1344\nodesize\log^2(\nodesize)$B.  
When $\nodesize = 1$ million, this data structure is roughly 500 GiB in size.


\paragraph{Using existing \boldmath $\ell_0$-samplers offers no advantage on modern hardware.}
The goal of a streaming connected components algorithm is to use smaller space than would be required to store the entire graph explicitly.  As we demonstrate empirically in Section ~\ref{sec:experiments}, the most space-efficient dynamic graph processing system, Aspen, requires roughly 4B of space for each edge in the graph.  A straightforward back-of-the-envelope calculation reveals that even for dense graphs with average degree $\nodesize/2$, \algname would use less space than Aspen only on very large inputs which require enormous RAM capacities and decades of processing time: $1344\nodesize\log^2(\nodesize)$B $\leq 4B \cdot \nodesize^2/4$ only when $\nodesize \geq 5\cdot 10^5$. Processing a half a million-node graph using \algname would require 220 GB of RAM and, at an ingestion rate of less than 35 edges per second, would take more than 56 years to process the graphs' roughly 64 billion edges.  While \algname's space complexity is much smaller than explicit graph representations like Aspen's asymptotically, in absolute terms it offers no advantage on modern hardware.

\subsection{Improved \boldmath  $\ell_0$-Sampler for Graph Connectivity}
\label{subsec:bettersketch}
We present \sketchname, an $\ell_0-$sampling algorithm for vectors on the integers mod 2, which is smaller than the best existing general-purpose $\ell_0$-sampling algorithm and is asymptotically faster to update. 
Since addition of characteristic vectors (Section \ref{subsec:alg}) can be thought of as addition over vectors $\in \mathbb{Z}_2$, \sketchname is sufficient for solving the connected components problem. Additionally, \sketchname may be useful for other sketching algorithms for problems such as edge- or vertex-connectivity, testing bipartiteness, and finding minimum spanning trees and densest subgraphs~\cite{Ahn2012,GuhaMT15,Ahn2012_2,10.1007/978-3-662-48054-0_39}.

\newtext{
Since \sketchname's goal is to recover a nonzero entry from vectors of integers mod 2, it can use a much simpler bucket data structure than the general-purpose $\ell_0$-sketch, improving space and update time costs.  The \sketchname algorithm is summarized in Figure \ref{fig:cube_code}. Each bucket $\indexsubset_{i,j}$ maintains 2 values: $\alpha_{i,j}$, which is used to recover the position of a single nonzero entry, and $\gamma_{i,j}$, which is used as a checksum.
}
$\alpha_{i,j}$ and $\gamma_{i,j}$ are each $O(\log(n))$ bits, and therefore require $O(1)$ machine words.
Since each vector value is either 0 or 1, $\Delta = 1$ for every stream update $(e, \Delta)$, and so for simplicity we refer to the update as $(e)$.

\begin{figure}[!t]
\begin{small}
\begin{algorithmic}[1]
\Function {update\_sketch}{idx}
\Comment{Toggle vector index `idx'}
    \ForAll {col $\in [0, q\log(1/\delta)$}
        \State col\_hash $\gets$ hash$_1$(col, idx)
        \State row $\gets 0$
        \State checksum $\gets$ hash$_2$(col, idx)
        \While {row == 0 OR col\_hash[row-1] == 0}
            \State col[row].$\alpha$ $\gets \text{col[row].$\alpha$} \oplus \text{idx}$
            \State col[row].$\gamma$ $\gets \text{col[row].$\gamma$} \oplus \text{checksum}$
            \State row $\gets \text{row} + 1$
        \EndWhile
    \EndFor
\EndFunction

\Function{query\_sketch}{ }
\Comment{Get a non-zero vector index}
    \ForAll {col $\in [0, q\log(1/\delta))$}
        \ForAll {bkt $\in$ col}
            \If {bkt.$\gamma$ == hash$_2$(col, bkt.$\alpha$)}
                \State \textbf{return} bkt.$\alpha$
                \Comment{Found a good bucket, done}
            \EndIf
        \EndFor
    \EndFor
    \State \textbf{return} sketch\_failure
    \Comment{All buckets bad}
\EndFunction
\end{algorithmic}
\end{small}

\setlength{\belowcaptionskip}{-12pt} 
\caption{Pseudocode for the \sketchname{} algorithm. }
\label{fig:cube_code}
\end{figure}

\newtext{
Function \textsc{update}\_\textsc{sketch}() in Figure \ref{fig:cube_code} describes how \sketchname processes a stream update. Given update $(e)$, if $h_1(e)\equiv0\pmod{2^i}$ then $e$ is in $\indexsubset_{i,j}$.  For each such $\indexsubset_{i,j}$, $\alpha_{i,j} = \alpha_{i,j} \oplus \bin(e)$ and $\gamma_{i,j} = \gamma_{i,j} \oplus h_2(\bin(e))$ where $\oplus$ denotes bitwise XOR, $\bin(e_w)$ denotes the binary representation of $e_w$, and $h_1$ and $h_2$ are hash functions drawn from a 2-wise independent family of hash functions. Note that the procedure for determining whether $e \in \indexsubset_{i,j}$ is identical to the algorithm in Figure \ref{fig:agm_code}, but the procedure for updating $\indexsubset_{i,j}$ is different. Importantly, \sketchname never performs modular exponentiation, which as we will show makes it a $log(\nodesize)$ factor faster than the existing algorithm in the average case. As a result of \textsc{update}\_\textsc{sketch}(), given a sequence of updates $(e_1), (e_2), \dots,  (e_k)$ to the data structure,
}

\begin{align}
\alpha_{i,j} &= \bigoplus_{w \in [k]} \bin (e_w) \label{eq:1} \\
\gamma_{i,j} &= \bigoplus_{w \in [k]} h_2(\bin (e_w)) \label{eq:2}
\end{align}

\newtext{
Function \textsc{query}\_\textsc{sketch}() describes how \sketchname returns a nonzero entry of the input vector. For any bucket $\indexsubset_{i,j}$:
}
$$
result = \left\{
        \begin{array}{ll}
            e' & \quad \text{ if } \alpha_{i,j} = \bin(e')\: \text{ and } \: \gamma_{i,j} = h_2(\bin(e')) \\
            \textsc{FAIL} & \quad \text{ if } \alpha_{i,j} = 0\: \text{ and } \: \gamma_{i,j} = 0 \text{ OR}\\
             & \quad \text{ if } \gamma_{i,j} \neq h_2(\alpha_{i,j})
        \end{array}
    \right.
$$


A nonzero entry is recovered from \sketchname by attempting to recover a nonzero entry from each $\indexsubset_{i,j}$ until one returns a value other than FAIL.  If no such bucket exists, the algorithm returns NULL.

\begin{theorem}
\sketchname is an $\ell_0$ sampler that, for input vector $x \in \mathbb{Z}_2^\veclength$, has space complexity $O(\log^2(\veclength)\log(1/\delta))$, worst-case update complexity $O(\log(\veclength)\log(1/\delta))$, average-case update complexity $O(\log(1/\delta))$, and failure probability at most $\delta$.
\end{theorem}

\begin{proof}
The space and update time results follow by construction: each bucket $\indexsubset_{i,j}$ requires a constant number of machine words, and $i \in [O(\log(\veclength)])$ and $j \in [O(\log(1/\delta)]$. Applying an update to any bucket $\indexsubset_{i,j}$ requires constant time, and in the worst case, an update will be applied to each of the $O(\log(\veclength)\log(1/\delta))$ buckets. In the average case an update is applied to $O(\log(1/\delta))$ buckets.

\begin{lemma}
\label{lem:goodset}
\sketchname's selection process succeeds with probability at least $1 - \delta$.  Equivalently, \sketchname contains a bucket $\indexsubset_{i,j}$ with a single nonzero entry, that is, $\prob{\exists i,j \text{ s.t. }\support{\indexsubset_{i,j}} = 1} \geq 1 - \delta.$
\end{lemma}

\begin{proof}
Adapted from~\cite{l0sketch}.  Choose $i \in [\log(n)]$ such that $2^{i-2} \leq \lVert x \rVert_0 < 2^{i-1}$ where $\lVert x \rVert_0$ denotes the $\ell_0$ norm of $x$, \ie the number of nonzero entries of $x$.  Let $A_x$ be the set of positions of nonzero entries in $x$. Then, since $h_1$ is drawn from a 2-universal family of hash functions, $\forall j \in [6\log(1/\delta)]$,
\begin{align*}
    \prob{\support{\indexsubset_{i,j} = 1}} &= \sum_{k \in A_x} \frac{1}{2^i}\left(1-\frac{1}{2^i}\right)^{\lVert x \rVert_0-1}\\
    &> \frac{\lVert x \rVert_0}{2^i}\left(1-\frac{\lVert x \rVert_0}{2^i}\right) > 1/8.
\end{align*}
Then $\prob{\support{\indexsubset_{i,j} \neq 1} \forall j \in [6\log(1/\delta)]} < (1 - 1/8)^{6\log(1/\delta)} = (7/8)^{6\log_{7/8}(1/\delta)/ \log_{7/8}(2)} = \delta^{-6/\log_{7/8}(2)} < \delta$.
\end{proof}

\begin{lemma}
\label{lem:goodcheck}
\sketchname's checksum succeeds with high probability.  That is, $\forall w,y,$ if$ \support{\indexsubset_{w,y}} = 1$ then $\gamma_{w,y} = h_2(\alpha_{w,y})$ and if $\support{\indexsubset_{w,y}} > 1$  then $\prob{\gamma_{w,y} \neq h_2(\alpha_{w,y})} \geq 1 - 1/\veclength^c$ for some constant c.

\end{lemma}

\begin{proof}
When $\indexsubset_{i,j}$ has a single nonzero entry, it always passes the error check.  That is, if $\support{\indexsubset_{i,j}} = 1$, $\alpha_{w,y} = \bin(e_i)$ where $e_i$ is the single nonzero element of $\indexsubset_{i,j}$, and $\gamma_{w,y} = h_2(\bin(e_i))$. 

When $\indexsubset_{i,j}$ has more than one nonzero entry, then it passes the error check only in the rare event of a hash collision:
If $\support{\indexsubset_{i,j}} > 1$, fix $e_i \in \indexsubset_{i,j}$. By equations (\ref{eq:1}) and (\ref{eq:2}), $\gamma_{w,y} = h_2(\alpha_{w,y})$ iff $\bigoplus_{j \in \indexsubset_{i,j}\setminus e_i} h_2(\bin(j)) \oplus h_2(\bin(e_i)) = h_2(\alpha_{w,y})$. Since $h_2$ is a 2-wise independent hash function, assuming that $\gamma_{i,j}$ is $c\log(\veclength)$ bits:

$$\prob{h_2(\bin(e_i)) = \left(\bigoplus_{j \in \indexsubset_{i,j}\setminus e_i} h_2(\bin(j)) \right) \oplus h_2(\alpha_{w,y})} = \frac{1}{2^{c\log(\veclength)}} = \frac{1}{\veclength^c}.$$
\end{proof}

Lemmas ~\ref{lem:goodset} and ~\ref{lem:goodcheck} imply that \sketchname is sampleable with probability $1 - \delta$ (see Definition ~\ref{def:l0}).  
\sketchname may be added via elementwise $\bigoplus$ (exclusive or). Linearity of \sketchname follows from the observation that exclusive or is a linear operation.
\end{proof}

Figure ~\ref{fig:l0speedup} illustrates that \sketchname is far faster than the standard $\ell_0$-sampling algorithm.  In fact, when sketching characteristic vectors of graphs with at least $10^5$ nodes, it is more than 3 orders of magnitude faster.  This dramatic speedup is a result both of \sketchname's asymptotically lower update time complexity, and the fact that its update cost is dominated by bitwise exclusive OR operations, which are in practice much faster than the division operations standard $\ell_0$-sampling performs. Finally, standard $\ell_0$ sampling is slowed significantly by the need to perform $O(\log(\nodesize)\log(1/\delta)$ modular exponentiation operations on 128-bit integers for each update when $\nodesize \geq 10^5$. \sketchname does not require 128-bit operations until processing graphs with tens of billions of nodes.

Figure ~\ref{fig:l0size} shows that, for the same input vector length and failure probability, \sketchname is twice as small as standard $\ell_0$ sampling for smaller vectors, and four times smaller for larger vectors.  This is a result of the fact that \sketchname's bucket data structures use half the machine words of standard $\ell_0$ sampling, and the fact that \sketchname does not need to use 128-bit integers for longer vectors.

\section{Buffering for I/O efficiency and improved parallelism}
\label{sec:io}
In the streaming connectivity problem, stream updates are \emph{fine-grained}: each update represents the insertion or deletion of a single edge.  Since streams are ordered arbitrarily, even a short sequence of stream updates can be highly non-local, inducing  changes throughout the graph.  As a result, \algname and similar graph streaming algorithms do not have good data locality in the worst case. This lack of locality can cause many CPU cache misses and therefore reduce the ingestion rate, even when sketches are stored in RAM. The cache-miss cost can be high since ingesting each stream update $(u,v,\Delta)$ requires modifying a logarithmic number of sketches, and can thus induce a poly-logarithmic number of cache misses.  The consequences are even worse if sketches are stored on disk since each edge update requires loading a logarithmic number of sketches from disk, leading to the following observation.

\begin{observation}
In the hybrid semi-streaming model with $\memsize = o(\nodesize \log^3(\nodesize))$ RAM and $D = \Omega(\nodesize \log^3(\nodesize))$ disk, \algname uses $\Omega(1)$ I/Os per update and processing the entire stream of length $\streamlength$ uses $\Omega(\streamlength) = \Omega(\edgesize)$ I/Os.
\end{observation}



Any sketching algorithm that scales out of core suffers severe performance degradation unless it amortizes the per-update overhead of accessing disk. Such an amortization is not straightforward, since sketching inherently makes use of hashing and as a result induces many random accesses, which are slow on persistent storage.  We now introduce a sketching algorithm for the streaming connected components problem that amortizes disk access costs, even on adversarial graph streams, and as a result is simultaneously a space-efficient graph semi-streaming algorithm and an I/O-efficient external-memory algorithm.  We also note that the design facilitates parallelism, which we experimentally verify in Section~\ref{sec:experiments}.

\subsection{I/O-Efficient Stream Ingestion}
\label{sec:gutter-tree-theory}

We describe \sysname's I/O efficient stream ingestion procedure in the hybrid streaming model (see Section \ref{subsec:models}).

Arbitrarily partition the nodes of the graph into \defn{node groups} of cardinality $\max\{1, \blocksize/\log^3(\nodesize)\}$.  Let $\nodegroup \subset \nodes$ denote a node group, and let $\sketch(\nodegroup)$ denote the node sketches associated with the nodes in $\nodegroup$. Store $\sketch(\nodegroup)$ contiguously on disk.  This allows $\sketch(\nodegroup)$ to be read into memory I/O efficiently: if node groups are of cardinality 1, then $\blocksize$ is smaller than the size of a node sketch, and if each node group has cardinality $\blocksize/\log^3(\nodesize) > 1$, then the sketches for the group have total size $O(\blocksize)$.

Applying stream update $((u,v),\Delta)$ to node sketches of $u$ and $v$ immediately upon  arrival takes $\Omega(1)$ I/Os since the corresponding sketches must be read from disk.  To amortize the cost of fetching sketches, \sysname only fetches $\sketch(\nodegroup_i)$ when it has collected $\max\{\blocksize, \log^3(\nodesize)\}$ updates for $\nodegroup_i$.  Since there may be $O(\nodesize)$ node groups, collecting these updates for each node group cannot be done in RAM. Instead, we collect these updates I/O efficiently on disk using a \defn{\treename}, a simplified version of a buffer tree~\cite{Arge95thebuffer} which uses $O(\nodesize(\log^3(\nodesize))$ space.

Like a buffer tree, a \treename consists of a tree whose vertices each have buffers of size $O(\memsize)$. Each non-leaf vertex has $O(\memsize/\blocksize)$ children.  We refer to a leaf vertex of the \treename as a \defn{gutter}, because it fills with stream data but is periodically emptied by applying the contained stream data to sketches.  Each leaf vertex in the \treename is associated with a node group $\nodegroup$ and has size $\max\{\blocksize, \log^3(\nodesize)\}$, the same size as $\sketch(\nodegroup)$.  When a gutter for node group $\nodegroup$ fills, \sysname reads $\sketch(\nodegroup)$ and the updates stored in the gutter into memory, applies the updates to $\sketch(\nodegroup)$, and writes $\sketch(\nodegroup)$ back to disk. Since data does not persist in leaf vertices, no rebalancing is necessary.

\begin{lemma}
\sysname's stream ingestion uses $O(\nodesize \log^3(\nodesize))$ space and $sort(\streamlength) = O(\streamlength/\blocksize(\log_{M/B}(\nodesize/\blocksize)))$ I/Os in the hybrid streaming setting.
\end{lemma}

\begin{proof}
\sysname's sketch data structures use $O(\nodesize \log^3(\nodesize))$ space.

Each leaf in the \treename has a gutter of size $\max\{\blocksize, \log^3(\nodesize)\}$.  This is one gutter for each node group and there are $\nodesize/ (\max\{1, \blocksize/\log^3(\nodesize)\})$ node groups so the total space for the leaves of the \treename is $O(\nodesize \log^3(\nodesize))$.

In the level above the leaves, there are $\nodesize \log^3(\nodesize)/\blocksize \cdot \blocksize/\memsize$ vertices each with size $\memsize$, so the total space used at this level is $O(\nodesize \log^3(\nodesize))$.  Each subsequent higher level of the tree uses $O(\memsize/\blocksize)$ space less than the level below it, so the total space used for the entire \treename is $O(\nodesize \log^3(\nodesize))$.

The I/O complexity of the \treename is equivalent to that of the buffer tree, except that leaf gutters are flushed by reading in the appropriate sketches from disk and applying the updates in the gutter to these sketches. Asymptotically this incurs no additional cost so the total I/O complexity for ingestion is $sort(\streamlength)$.
\end{proof}

\subsection{I/O-Efficient Connectivity Computation}

\begin{lemma}
Once all stream updates have been processed, \sysname computes connected components using $O((\nodesize \log^3(\nodesize)/\blocksize)  + (\nodesize \log^*(\nodesize)) $ I/Os in the hybrid streaming model.
\end{lemma}

\begin{proof}
Each round of Boruvka's algorithm has three phases.  In the first, an edge is recovered from the sketch of each current connected component.  In the second, for each edge its endpoints are merged in a disjoint set union data structure which keeps track of the current connected components.  In the third phase, for each pair of connected components merged in phase 2, the corresponding sketches are summed together.  We analyze the I/O cost of each phase of a round separately.

In the first round, to query the sketches in the first phase, all of the sketches must be read into RAM which can be done with a single scan. This uses $O(\nodesize \log^3(\nodesize)/\blocksize)$ I/Os.

The disjoint set union data structure has size $\O(\nodesize)$ and must be stored on disk. In the second phase the cost of each DSU merge is $\log^*(\nodesize)$ I/Os, so the total I/O cost is $\nodesize\log^*(\nodesize)$.

In the third phase, summing the sketches of the merged components together is I/O efficient if $\blocksize = O(\log^3(\nodesize))$, since the disk reads and writes necessary for summing sketches are the size of a block or larger. The cost for the third phase is $\O(\nodesize \log^3(\nodesize)/\blocksize)$.  

If $\blocksize = \omega(\log^3(\nodesize))$, sketches are much smaller than the block size. Since the merges performed in each round of Boruvka are a function both of the input stream and of the randomness of the sketches, these merges induce random accesses to the sketches on disk and so summing the sketches for each merge takes $O(1)$ I/Os. In total, the third phase takes $O(\nodesize)$ I/Os in this case.
\end{proof}

\begin{corollary}
When $\edgesize = \Omega(\nodesize \log^3(\nodesize))$ and $\blocksize = o(\log^3(\nodesize))$ or $\memsize = O(\nodesize)$, \sysname is I/O optimal for the connected components problem; \ie it uses $sort(\edgesize) = O(\edgesize/\blocksize(\log_{M/B}(\nodesize/\blocksize)))$ I/Os.
\end{corollary}


\newtext{
Note that for optimality the graph cannot be too sparse. In practice, for some graph streams $\memsize = O(\nodesize\blocksize)$ and $D = O(\nodesize \log^3(\nodesize))$.  In this case, we can omit the upper levels of the \treename and write I/O efficiently to the leaf gutters stored on disk.  In Section ~\ref{sec:system} we describe how \sysname can perform stream ingestion using either a full \treename or just the leaf gutters, and evaluate the performance of both approaches in Section ~\ref{sec:experiments}.
}


\section{System Description}
\label{sec:system}

\begin{figure}[!t]
  \centering
  \includegraphics[width=.5\textwidth]{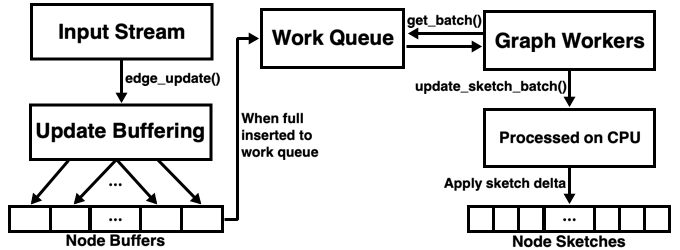}
\setlength{\belowcaptionskip}{-12pt} 
    \caption{\textbf{\sysname stream ingestion data flow.} }
    \label{fig:system}
\end{figure}

The \sysname algorithm is split into two components: \defn{stream ingestion}, in which edge updates are processed and stored using \sketchname, and \defn{query-processing}, in which a spanning forest for the graph is recovered from these sketches.  These components use SSD when the sketches are so large that they do not fit in RAM.  Their implementations are parallel for better performance on multi-core systems.

\begin{figure}[!t]
  \input{figures/api}
  \caption{Pseudocode for \sysname's core stream ingestion routines. \edgeupdate{} is part of the user API, while \textsc{do\_batch\_update()}, and \textsc{update\_sketch\_batch()} are internal functions.}
  \label{fig:api}
\end{figure}
\begin{figure}[!t]
  \input{figures/api2}
  
\setlength{\belowcaptionskip}{-16pt} 
  \caption{Pseudocode for \sysname's core post-processing routines. \listspanningforest{} is part of the user API, while \textsc{cleanup()} is an internal function.}
  \label{fig:api2}
\end{figure}

\sysname's user-facing API consists of \edgeupdate{} for processing stream updates, and \listspanningforest{} to compute and return the connected components.   On initialization \sysname allocates $\log(\nodesize)$ \sketchname data structures for each node in the graph, for a total sketch size of approximately $280\nodesize\cdot\log^2(\nodesize)$ bytes. It also initializes its buffering data structure.

\subsection{Stream Ingestion}  
Each update in the input stream is immediately placed into a buffering system.  Periodically, the buffering system produces a \batch{} of updates bound for the same graph node $u$.  This batch is inserted into a work queue, which then hands the batch off to a \defn{Graph Worker}, i.e., 
a thread for carrying out batched sketch updating.  Because each batch is only applied to a single node sketch, and because each of the $\log(\nodesize)$ $\sketchname$es in a node sketch can be updated in parallel, many Graph Workers can operate in parallel without contention (see Section \ref{subsec:multithreading}).
A high-level illustration and pseudo code of \sysname stream ingestion are shown in Figure~\ref{fig:system} and Figure~\ref{fig:api} respectively.




\paragraph{Buffering.}
\sysname's buffering system ingests updates from the stream and periodically outputs a batch of updates for a single node in the graph. 
\sysname implements two buffering data structures: a \treename{}, described in Section~\ref{sec:io}, and a simplified version of the \treename, which only includes the leaves. Depending upon available memory, \sysname uses only one of these two buffering structures at any time. 
\newtext{
The leaf-only version is fundamentally a special case gutter tree used when sufficient memory is available $(\memsize > \nodesize\cdot\blocksize)$ and is optimized for this case.
}

These buffering techniques confer several benefits.
\newtext{First, when \sysname's sketches are so large that they do not fit in RAM and are stored on SSD, applying updates to a single node sketch in large batches amortizes the I/O cost of reading the node sketch into memory. Without buffering, each stream update would incur $\Omega(1)$ I/Os in the worst case. 
We demonstrate in Section ~\ref{sec:multi-thread-exp} that buffering facilitates I/O efficiency and parallelism.}




\subparagraph{Gutter tree.}
\sysname allocates $8\textsc{MB}$ for each non-leaf buffer in the \treename. 
The \treename writes updates to the disk in blocks of $16 \textsc{KB}$, and has a fan-out of $\frac{8\textsc{MB}}{16\textsc{KB}}=512$.
\newtext{
A write block of $16 \textsc{KB}$ is an efficient I/O granularity for SSDs and a buffer size of $8\textsc{MB}$ balances buffering performance with the latency of flushing updates through the \treename{}.
}
When $\nodesize > 5 \cdot 10^4$, the size of a sketch is greater than $100 \textsc{KB}$, much larger than the $16\textsc{KB}$ block.
Therefore, the leaf nodes of the \treename{} accumulate updates for a single graph node. \sysname allocates space for each leaf gutter equal to twice the size of a node sketch.


When we initialize \sysname{}, we leverage the static structure of the \treename to pre-allocate its disk space. 
A call to \bufferinsert{$\{u,v\}$} inserts $\{u,v\}$ to the root buffer of the \treename.
Another thread asynchronously flushes the contents of full buffers to the appropriate child using the \emph{pwrite} system call.
When a flush causes the buffer of a child node to fill, that child node is recursively flushed before the flush of the parent continues. 
When a leaf gutter is full this thread moves the \batch{} of updates into the work queue for processing by Graph Workers in \dobatchupdate{}.
\vspace{-.5em}

\subparagraph{Leaf-only gutter tree.} 
For each graph node $u$ we maintain a gutter that accumulates updates for $u$.
When the system is initialized, we allocate the memory for each of these gutters.
By default, each leaf gutter is 1/2 the size of a node sketch. 
\newtext{
This choice balances RAM usage with I/O efficiency as shown in Section~\ref{sec:multi-thread-exp}.
}

\bufferinsert{$(u,v)$} inserts edge $e=(u,v)$ directly into the gutter for node $u$. 
As before, when the gutter becomes full, it is flushed and the batch is inserted into the work queue.

Note that the leaf-only gutter data structure need not fit entirely in RAM, so long as at least a page of memory is available per buffer the rest can be efficiently swapped to SSD; see Section~\ref{sec:experiments}.

\paragraph{Work queue.}
The work queue functions as a simple solution for the producer-consumer problem, in which the thread filling buffers produces work and the Graph Workers consume it.
Once a buffer is filled the \textsc{buffer\_insert}() function inserts the \batch{} of updates into the work queue.  
Later a Graph Worker removes the \batch{} from the front of the queue in \dobatchupdate{}.

Insertions to the queue are blocked while the queue is full, and Graph Workers in need of work are blocked while the queue is empty. The work queue can hold up to $8g$ batches, where $g$ is the number of Graph Workers. 
\newtext{
A moderate work queue capacity of $8g$ limits the time either the buffering system or graph workers spend waiting on the queue, even when batch creation is volatile, while keeping the memory usage of the work queue low.
}



\paragraph{Sketch updates}.
In each call to \dobatchupdate{}, Graph Workers call \getbatch{} to receive a batch of updates bound for a particular node $u$ from the work queue. The Graph Worker then uses \sketchbatch{sketch$_u$, batch} to update each of the $O(\log(\nodesize))$ \sketchnames in the node sketch of $u$.


As described in Section~\ref{subsec:bettersketch}, a \sketchname is a vector of buckets, each of which consists of a 64 bit $\alpha$ value and a 32 bit $\gamma$ value.
Each \sketchname stores a two dimensional array $A$ of buckets $\indexsubset_{i,j}$, with dimensions $\log(\nodesize^2) \times (\log(1/\delta) = 7)$.
To apply an update $(e=\{u,v\})$ to a \sketchname, the Graph Worker determines which buckets $\indexsubset_{i,j}$ contain $e$, and sets $\alpha_{i,j}:=\alpha_{i,j}\oplus e$ and $\gamma_{i,j}:=\gamma_{i,j}\oplus h_y(e)$. The hash values are calculated using xxHash~\cite{xxhash}.

Each \sketchname data structure uses $7\log(\nodesize^2) = 14\log(\nodesize)$ 12B buckets.
In total, this is $168\log(\nodesize)$ bytes per \sketchname, and $168\log(\nodesize)\log_{2/3}(\nodesize)$ bytes per node sketch.



\paragraph{Multithreading sketch updates}.
\label{subsec:multithreading}
Applying a \batch to a node sketch in \dobatchupdate{} is handled asynchronously by a Graph Worker, allowing what we call \defn{\batch{}-level parallelism}. We implement these workers using \verb|C++| STL threads.

We use \verb|OpenMP|~\cite{ompAPI} to dispatch a group of threads to process each \sketchname update in \sketchbatch{}. We refer to this as \defn{sketch-level parallelism}. \verb|OpenMP| allows us to specify the number of threads to allocate to a task and handles work allocation transparently. When updating a node sketch, applying a \batch{} to each \sketchname is treated as one work unit and \verb|OpenMP| allocates the $\log(V)$ units between the apportioned threads.

Implementing both \batch{}- and sketch-level parallelism gives us a natural way to tune \sysname's performance. For instance, we can decide to configure more Graph Workers with fewer threads per group, or fewer Graph Workers with more threads per group. 
We experimentally determine a good configuration for our hardware and datasets (see Section~\ref{sec:multi-thread-exp}).

A single work unit is never shared between threads in the same group. As a result, a \sketchname is only modified by one thread in a group, so no locking is necessary at the sketch level. However, locking is necessary at the \batch{} level because consecutive \batch{} updates may be requested to the same node sketch, and thus multiple graph workers may seek to dispatch thread groups to the same sub-sketches. We minimize the size of this critical section by exploiting linearity of $\ell_0$-samplers. Rather than locking a node sketch $S(x)$ for the entire batch operation, we apply the updates to an empty sketch $S(x_0)$ and lock only to add $S(x) = S(x) + S(x_0)$.
\subsection{Query Processing} 

When a connectivity query is issued, \sysname calls \listspanningforest{} which returns a spanning forest of the graph. The first step of post-processing is to flush the buffering data structure of any remaining updates, moving the \batch{}es to the work queue in \cleanup{}. We then wait for the Graph Workers to finish processing these \batch{}es. Finally, \sysname runs Boruvka's algorithm to generate a spanning forest of the input graph.

\section{Evaluation}
\label{sec:experiments}

\paragraph{Experimental setup.}
We implemented \sysname as a C++14 executable compiled with g++ version 9.3 for Ubuntu. All experiments were run on a Dell Precision 7820 with 24-core 2-way hyperthreaded Intel(R) Xeon(R) Gold 5220R CPU @ 2.20GHz, 64GB 4x16GB DDR4 2933MHz RDIMM ECC Memory and two 1 TB Samsung 870 EVO SSDs. In some of our experiments we artificially limited RAM to force systems to page to disk using Linux Control Groups. We put a swap partition and the gutter tree data on one of the two SSDs, and the other SSD held the datasets.

\subsection{Datasets}
\label{subsec:datasets}
We used two types of data sets in this paper.  First, we generated large, dense graphs using a Graph500 specification, and converted these to streams for our evaluation.  
We also evaluated correctness on graphs from the SNAP graph repository~\cite{snapnets} and the Network Repository~\cite{netrepo}. All data sets used are described in \tabref{datasets}.

\subparagraph{Synthesizing Dense Graphs and Streams} \label{sec:synth} We created undirected graphs using the Graph500 Kronecker generator.  We produced five simple, undirected graphs.  These graphs are dense: each has roughly one half of all possible edges. The Graph500 generator does not output simple graphs by default, so to produce our five simple graphs we pruned duplicate edges and self-loops~\cite{Ang2010IntroducingTG}.
 

\begin{figure}
\begin{center}
\def\arraystretch{1.1}
\begin{tabular}{ |c|c|c|c| } 
 \hline
 Name & \# of Nodes & \# of Edges & \# Stream Updates\\ 
 \hline
 \textbf{kron13} & $2^{13}$ & $1.7\times 10^7$ & $1.8\times 10^7$\\
 \hline
 \textbf{kron15} & $2^{15}$ & $2.7 \times 10^8$ & $2.8 \times 10^8$\\ 
 \hline
 \textbf{kron16} & $2^{16}$ & $1.1 \times 10^9$ & $1.1\times 10^9$\\ 
 \hline
 \textbf{kron17} & $2^{17}$ & $4.3 \times 10^9$ & $4.5\times 10^9$\\ 
 \hline
 \textbf{kron18} & $2^{18}$ & $1.7 \times 10^{10}$ & $1.8 \times 10^{10}$ \\ 
 \hline
 \textbf{p2p-gnutella} & $6.3\times 10^4$ & $1.5\times 10^5$ & $2.9 \times 10^5$ \\ 
 \hline
 \textbf{rec-amazon} & $9.2\times 10^4$ & $1.3\times 10^5$ & $2.5 \times 10^5$ \\ 
 \hline
 \textbf{google-plus} & $1.1 \times 10^5$ & $1.4 \times 10^7$ & $2.7\times 10^7$ \\ 
 \hline
 \textbf{web-uk} & $1.3 \times 10^5$ & $1.2 \times 10^7$ & $2.3\times 10^7$ \\  
 \hline
 
\end{tabular}
\end{center}
\setlength{\belowcaptionskip}{-16pt} 
\caption{Dimensions of datasets used in this evaluation. 
}
\label{tab:datasets}
\end{figure}

We then transformed each of the 5 graphs into a random stream of edge insertions and deletions with the following guarantees: (i) an insertion of edge $e$ always occurs before a deletion of $e$, (ii) an edge never receives two consecutive updates of the same type, (iii) we disconnect a small (fewer than 150) set of nodes from the rest of the graph, and (iv) by the end of the stream, exactly the input graph (with the exception of the edges removed to disconnect the vertices in (iii)) remains. Note that this mechanism deliberately adds edges not in the original graph, but they are always subsequently deleted. We implemented (iii) to guarantee some non-trivial connected components in each stream's final graph. 

\subparagraph{Publicly Available Datasets}
We also used the following real-world data sets. \textbf{p2p-gnutella} is a graph representing the Gnutella peer-to-peer network~\cite{ripeanu2002mapping}. \textbf{rec-amazon} is a co-purchase recommendation graph for products listed on Amazon~\cite{leskovec2007dynamics}, where each node represents a product and there is an edge between two nodes if their corresponding products are frequently purchased together. \textbf{google-plus} is a graph among users of the Google Plus social network~\cite{mcauley2012learning} where edges represent follower relations. \textbf{web-uk} is a web graph, where edges represent links between pages~\cite{netrepo}. Each of these real-world graphs was converted to a stream using the process described above.


\subsection{\sysname is Fast and Compact}

We now demonstrate that, given the same memory resources, \sysname can handle larger inputs than Aspen and Terrace on sufficiently large and dense graph streams. \newtext{We also show that unlike these systems, \sysname maintains good performance when its data structures are stored on SSD.}

Both Aspen and Terrace are optimized for the \defn{batch-parallel} model of dynamic graph processing. In this model, updates are applied to a non-empty graph in batches containing exclusively insertions or exclusively deletions. This contrasts with our streaming model, an initially empty graph is defined entirely from a stream of interspersed inserts and deletes. 
To avoid unfairly penalizing Aspen and Terrace, we group the input stream into batches insertions and deletions to these systems (ignoring any query correctness issues this may introduce) and present these batches as the input stream. 
Whenever one of these arrays fills, we feed it into the appropriate batch update function provided by Aspen or Terrace\footnote{Note that Terrace does not currently support batch deletions, so we rely on its individual edge deletion functionality instead and do not maintain a deletions array.}. 
We ran \sysname, Aspen, and Terrace on each Kronecker stream. We used a batch size of $10^6$ for Aspen and Terrace because we found this to produce the highest ingestion rates for both systems. To record memory usage we logged the output of the Linux top command tracking each system every five seconds. 
All experiments were run for a maximum of 24 hours.

\subparagraph{Memory profiling.}
\sysname's space-efficient \sketchnames make it a ${O}(\nodesize/ \log^3(\nodesize))$-factor smaller than Aspen or Terrace asymptotically.  Given the polylogarithmic factors and constants, this experiment determines the actual crossover point where \sysname is more compact than Aspen and Terrace.
As shown in Figure \ref{fig:size}, \sysname is smaller than Terrace even on kron15, and smaller than Aspen on kron17 and kron18. 


\begin{figure}[!t]
  \centering
  
  \begin{subfigure}[b]{0.5\textwidth}
    \centering
    \begin{tabular}{ |c|c|c|c| } 
      \hline
      Dataset & Aspen & Terrace & \sysname\\ 
      \hline
      kron13 & $0.000328$ & $0.000519$ & $0.58$\\ 
      \hline
      kron15 & $3.40$ & $6.30$ & $3.10$\\ 
      \hline
      kron16 & $6.40$ & $23.7$ & $7.00$\\ 
      \hline
      kron17 & $16.8$ & $96.0$* & $15.7$\\ 
      \hline
      kron18 & $57.7$ & N/A & $35.1$\\ 
      \hline
    \end{tabular}
    \caption{Space used by each system. All numbers listed in GiB. Terrace did not finish processing kron17 within 24 hours.}
    \label{tab:size}
  \end{subfigure}
  
  \hfill
  
  \begin{subfigure}[b]{0.5\textwidth}
    \centering
    \includegraphics[width=\textwidth]{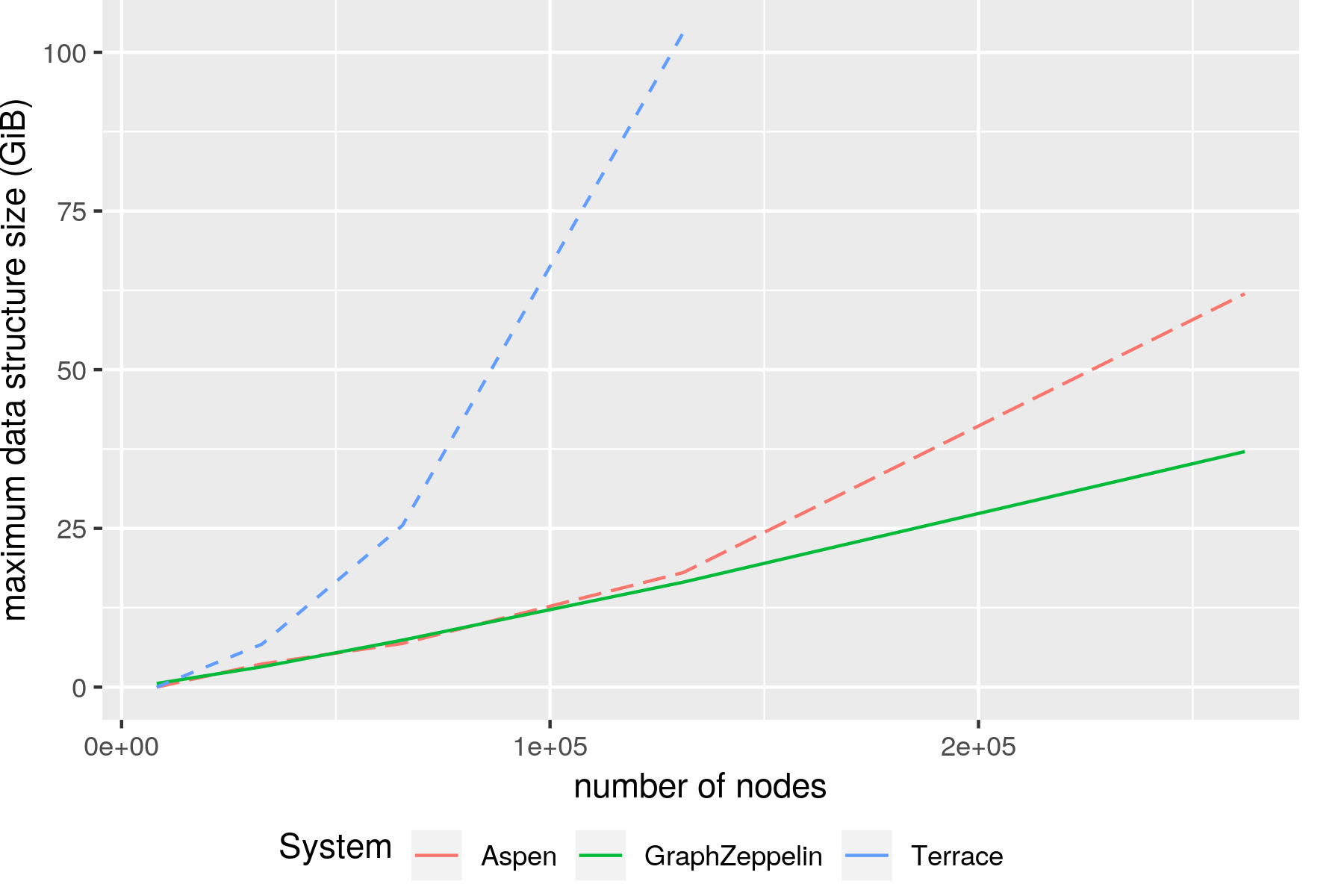}
\setlength{\belowcaptionskip}{-16pt} 
    \caption{
\sysname is asymptotically more memory efficient than either Aspen or Terrace on large, dense graphs. 
    }\label{fig:size_plot}
  \end{subfigure}
  \setlength{\belowcaptionskip}{-16pt} 
  \caption{\sysname uses less space than Aspen or Terrace to process large, dense graph streams.}
  \label{fig:size}
\end{figure}

\subparagraph{I/O Performance and Ingestion Rate.}
\label{sec:speed-comparison}
 Unlike Aspen and Terrace, \sysname maintains consistently high ingestion rates when its data structures are stored on SSD. \newtext{In Figure~\ref{fig:speed_disk} we summarize the results of running Aspen, Terrace, and \sysname with only 16GB of RAM.}  The ingestion rates of both Aspen and Terrace plummet once their data structures exceed 16GB in size and they are forced to store excess data on SSD. 
 Neither Aspen nor Terrace were able to finish their largest evaluated stream within 24 hours ($2^{17}$ for Terrace and $2^{18}$ for Aspen).
In comparison, \sysname's ingestion rate remains high when its memory consumption extends into secondary storage. \sysname's \treename{} finished the kron18 stream with an average ingestion rate of 2.50 million updates per second, a 29\% reduction to its performance compared to when its sketches are stored entirely in RAM.

\begin{figure}[!t]
\centering

\begin{subfigure}[b]{0.5\textwidth}
\centering
\begin{tabular}{ |c|c|c|c|c| } 
 \hline
 Dataset & Aspen & Terrace & Gutter Tree \textsc{GZ} & Leaf-Only \textsc{GZ} \\ 
 \hline
 kron13 & $4.98$ & $0.138$ & $3.93$ & $5.22$\\ 
 \hline
 kron15 & $3.54$ & $0.133$ & $3.77$ & $4.87$\\ 
 \hline
 kron16 & $2.54$ & $0.0143$ & $3.59$ & $4.51$\\ 
 \hline
 kron17 & $1.90$ & $0.0404$* & $3.26$ & $4.24$\\ 
 \hline
 kron18 & $0.0759$* & N/A & $2.50$ & $2.49$\\ 
 \hline
\end{tabular}
\caption{Ingestion rates in millions of updates per second.
Asterisks indicate the system did not finish within 24 hours.}\label{tab:speed}
\end{subfigure}

\hfill

\begin{subfigure}{0.5\textwidth}
  \centering
  \includegraphics[width=\textwidth]{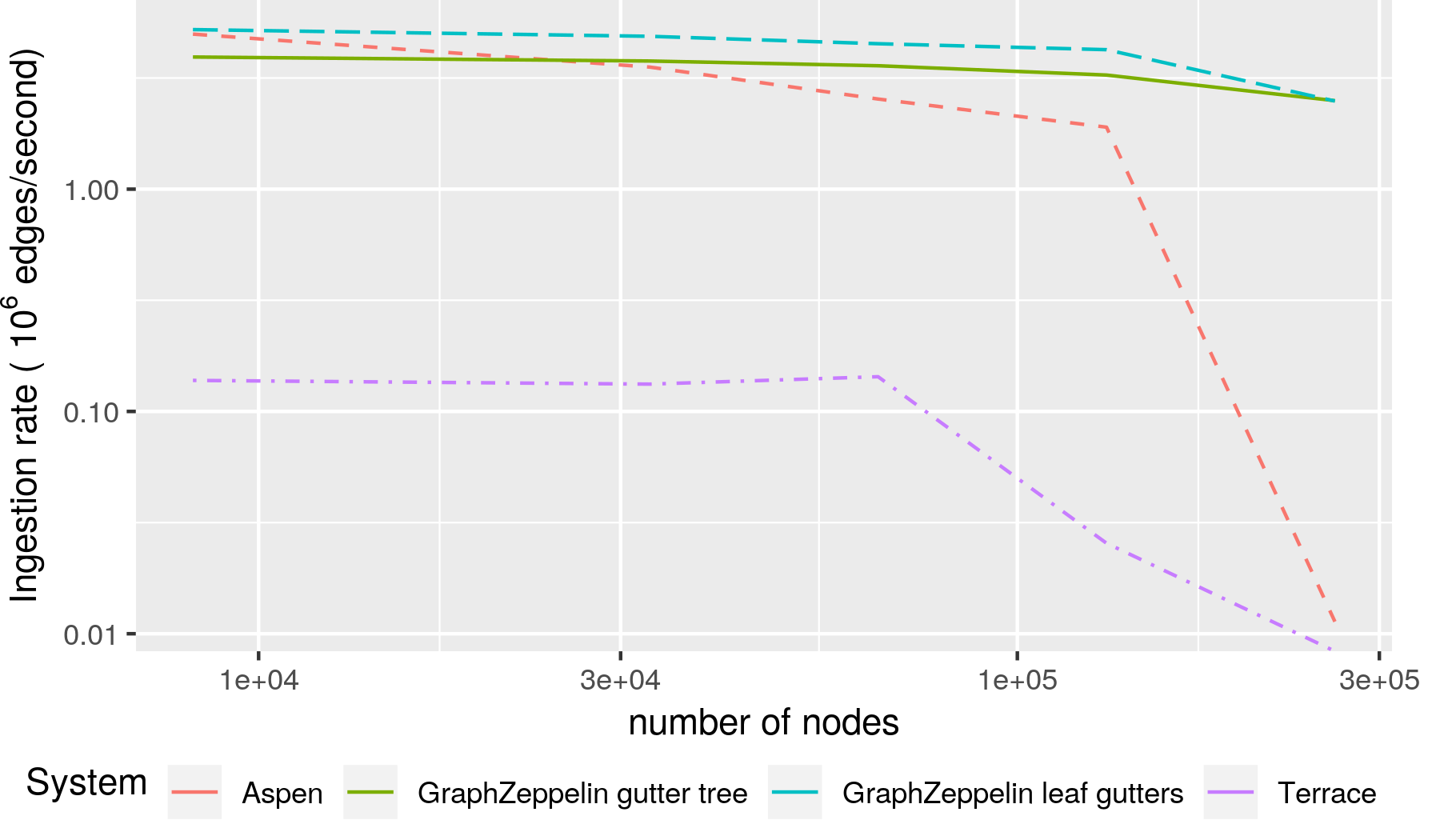}
\caption{Aspen and Terrace perform poorly on disk while \sysname's stream ingestion rate remains high.} \label{fig:speed_disk}
\end{subfigure}

\hfill

\begin{subfigure}[b]{0.5\textwidth}
\begin{tabular}{ |c|c|c|c|c| } 
 \hline
 Dataset & Aspen & Terrace & Gutter Tree & Leaf-Only Gutters \\ 
 \hline
 kron13 & $0.041$ & $0.126$ & $0.02$ & $0.02$\\ 
 \hline
 kron15 & $0.202$ & $0.800$ & $0.10$ & $0.10$\\ 
 \hline
 kron16 & $0.746$ & $1.260$ & $0.22$ & $0.19$\\ 
 \hline
 kron17 & $3.11$ & N/A & $0.44$ & $ 0.42$\\
 \hline
 kron18 & N/A & N/A & $97.5$ & $103$\\ 
 \hline
\end{tabular}
\caption{CC computation time after stream ingestion.
}
\label{tab:cc_time}
\end{subfigure}

\caption{\sysname remains fast even when its data structures are stored on disk, unlike Aspen and Terrace.}
\label{fig:IO_throughput}
\end{figure}

\begin{figure}[!t]
  \centering
  \includegraphics[width=.5\textwidth]{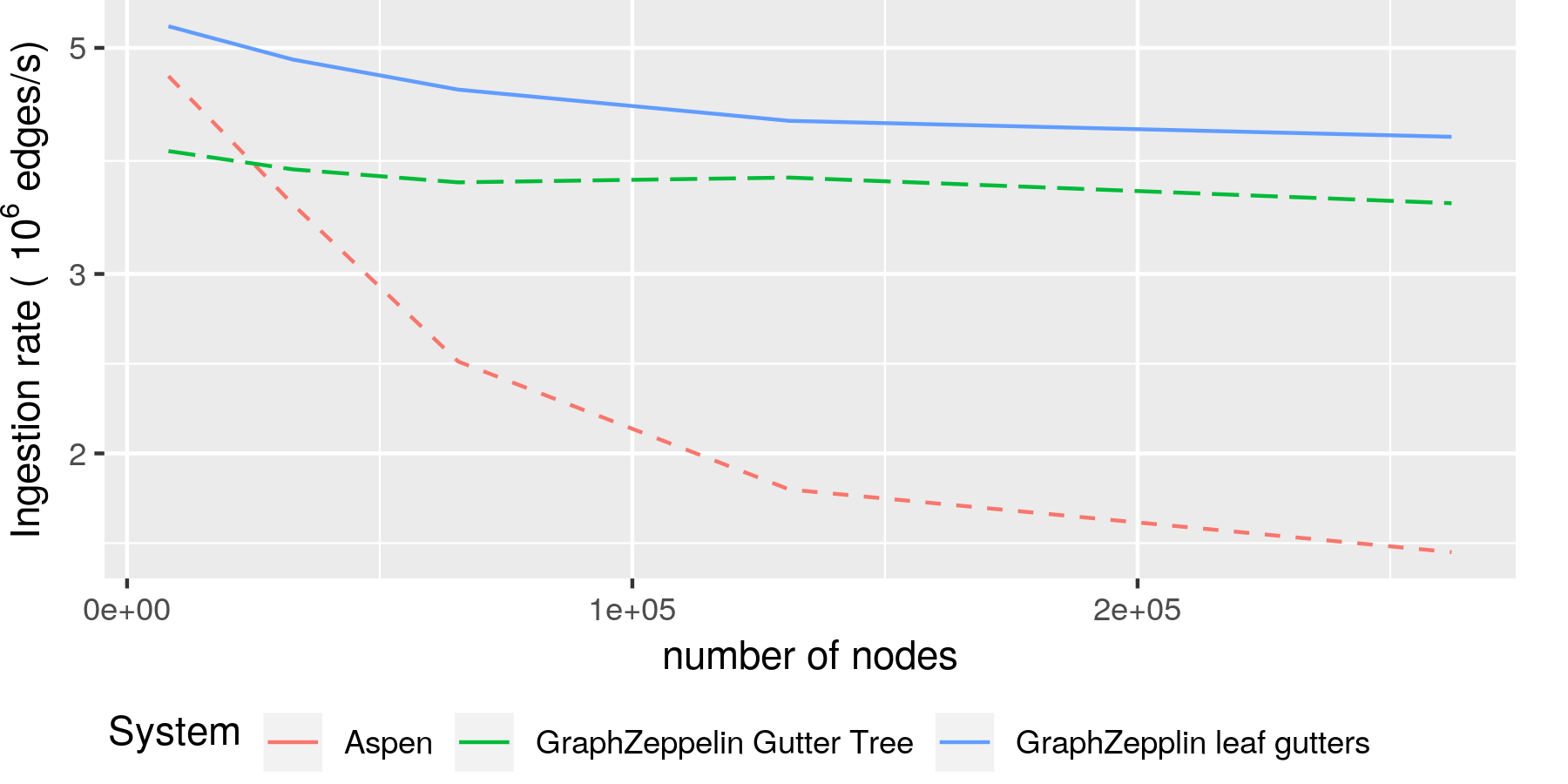}
\setlength{\abovecaptionskip}{-8pt} 
\setlength{\belowcaptionskip}{-12pt} 
\caption{\sysname is faster than Aspen and Terrace even when all data structures fit in RAM. 
} \label{fig:speed_mem}
\end{figure}


In RAM, \sysname's ingestion rate is higher than Aspen's on all Kronecker streams. We summarize these results in  \figref{speed_mem}.  Notably, on kron18 \sysname ingests 4.09 million updates per second, nearly three times faster than Aspen. \sysname ingests more than an order of magnitude faster than Terrace on these streams, so we omit it from the figure.

\subsection{\sysname is Reliable}
\sysname's sketching algorithm is not deterministically correct: it has a nonzero failure probability, which is guaranteed to be at most $1/\nodesize^c$ for some constant $c$. To establish that failures do not occur in practice, we compared \sysname with an in-memory adjacency matrix stored as a bit vector. Specifically, we applied stream updates to \sysname and the adjacency matrix and periodically queried \sysname and compared its results with the output of running Kruskal's algorithm on the adajacency matrix. We performed 1000 such correctness checks each on the kron17, p2p-gnutella, rec-amazon, google-plus, and web-uk streams. No failures were ever observed. While our algorithm's performance is optimized for dense graphs, this experiment demonstrates that it succeeds with high probability for both dense and sparse graphs.


\subsection{\sysname is Highly Parallel}
\label{sec:multi-thread-exp}
Due to the atomized nature of sketch updates, we expect stream ingestion to scale well on multi-core systems. We experimentally demonstrate this claim by varying the number of threads used for processing updates and observe a significant speed-up.

\begin{figure}
    \centering
    \includegraphics[width=.5\textwidth]{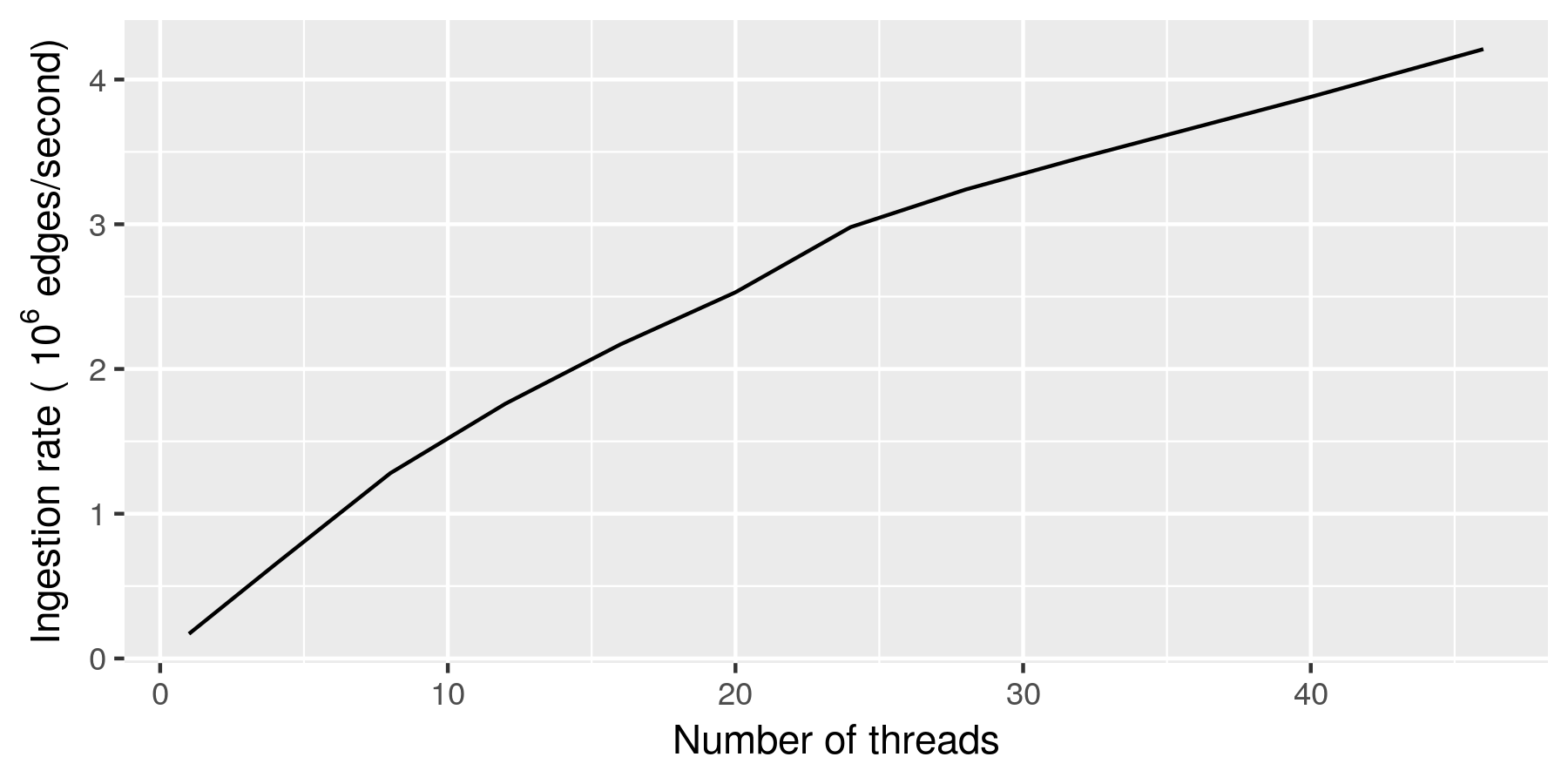}
    \caption{\sysname updates sketches in parallel, increasing ingestion rate by $26\times$ when using 46 threads.}
    \label{fig:parallel}
\end{figure}

Figure~\ref{fig:parallel} shows the ingestion rate of \sysname as the number of threads processing the kron17 graph stream increases. The threads are given a pool of $64\textsc{GB}$ RAM so that the parallel performance can be measured without memory contention. To avoid external memory accesses, we use leaf-only gutters for buffering. 
The per-thread increase in ingestion rate is significant; the ingestion rate for 46 threads is approximately 26 times higher than that of a single thread. 
Additionally, at 46 threads the marginal ingestion rate is still positive, suggesting that adding more threads would further increase performance.

We also experimentally determined that a group size of one gives the best performance with our combination of machine and inputs. 


\subsection{\sysname Buffering Facilitates Parallelism and I/O Efficiency}
\label{sec:buffer-size-exp}
\begin{figure}
    \centering
    \includegraphics[width=.5\textwidth]{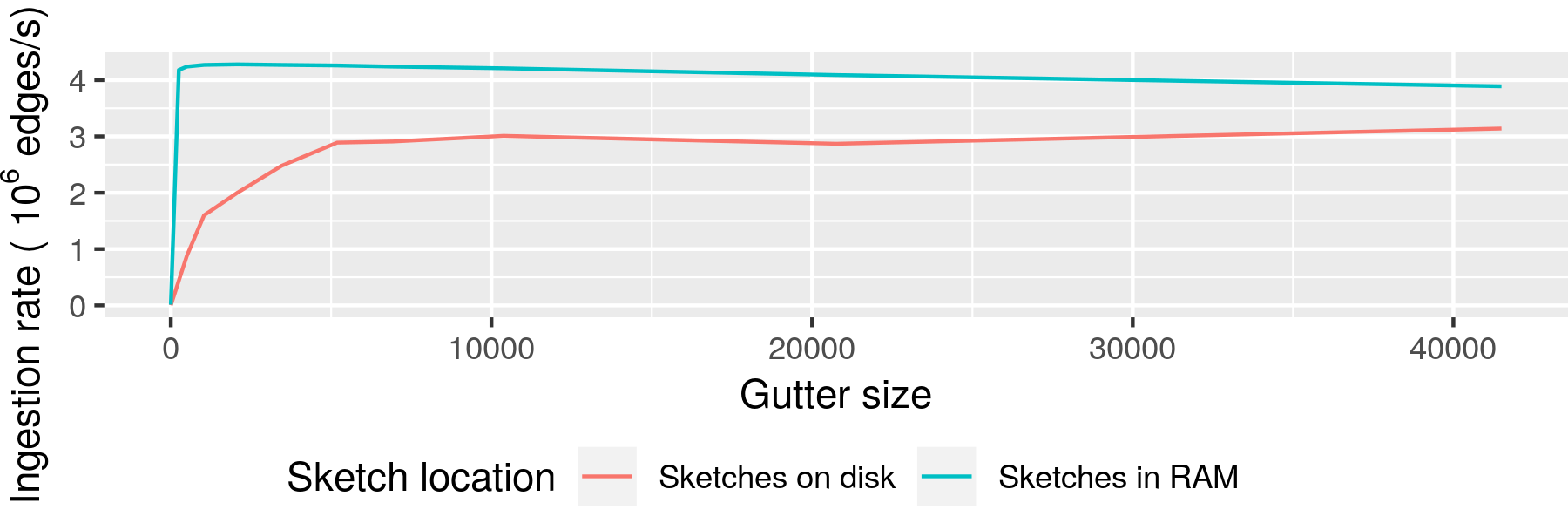}
    \setlength{\abovecaptionskip}{-10pt}
    \setlength{\belowcaptionskip}{-12pt} 
    \caption{\sysname gutter size vs ingestion speed. 
    }
    \label{fig:gutter_size}
\end{figure}
\vspace{-.5em}

Applying sketch updates is highly scaleable, but only if updates are buffered and applied in \batch{}es.  When sketches are stored on disk, processing each update individually requires $\Omega(1)$ IOs. Additionally, cache contention and thread synchronization bottleneck the ingestion rate even when sketches are in RAM. For these reasons we retain buffers of a constant factor $f$ of the node-sketch size.

Figure~\ref{fig:gutter_size} summarizes the ingestion rate of \sysname on the kron17 stream for different values of $f$ when the sketches are stored in RAM and when they are stored on disk. \sysname is given 46 Graph Workers and a group size of 1. With buffers of size 1 (no buffering), \sysname ingests 130,000 updates per second in RAM, 33 times slower than when $f = .10$. On SSD, the ingestion rate is only 2000 insertions per second, 3 orders of magnitude slower than peak on-disk performance.

When the sketches fit in RAM, performance increases rapidly indicating that $f$ can be quite small while providing a high ingestion rate. However, once memory requirements exceed main memory, $f$ must be larger to offset disk IOs. To achieve an ingestion rate within 5\% of peak performance on kron17, $f$ as small as $0.01$ is sufficient for entirely in RAM computation, while $f=.50$ is required when node sketches partially reside on disk.

\subsection{Connectivity Queries are Fast}

\begin{figure}[!t]
  \centering
  
\begin{subfigure}[b]{0.5\textwidth}
    \centering
    \includegraphics[width=\textwidth]{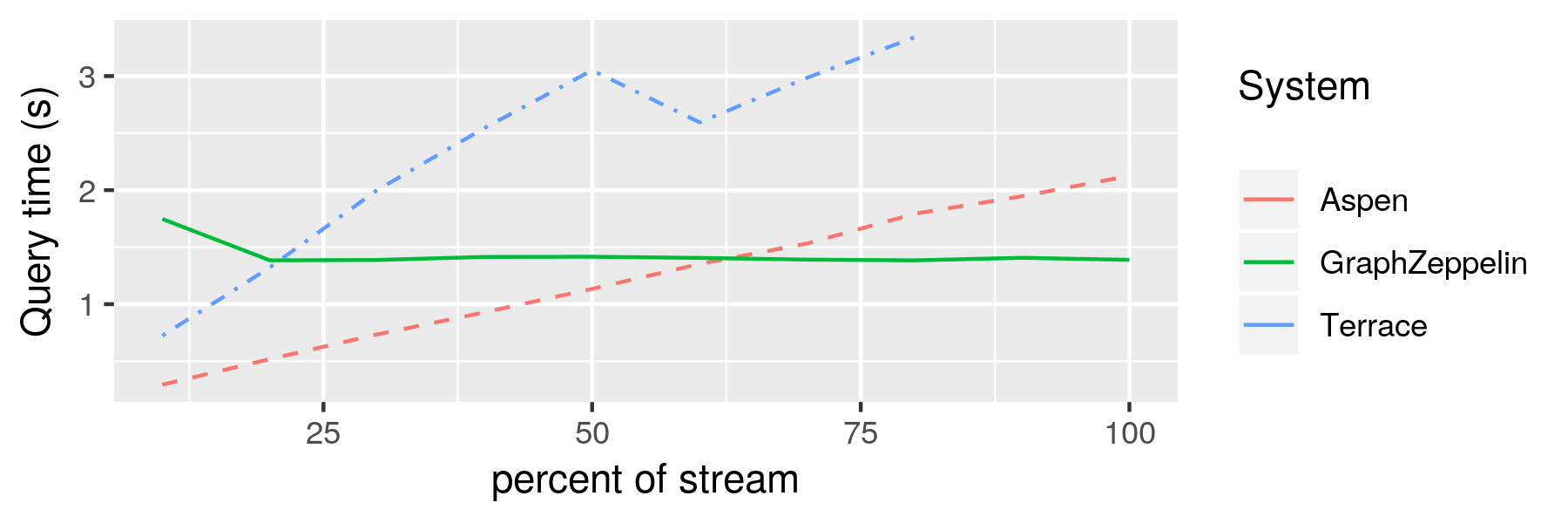}
    \caption{
In-memory query times. 
    }\label{fig:query1}
  \end{subfigure}

  \hfill
  
  \begin{subfigure}[b]{0.5\textwidth}
    \centering
    \includegraphics[width=\textwidth]{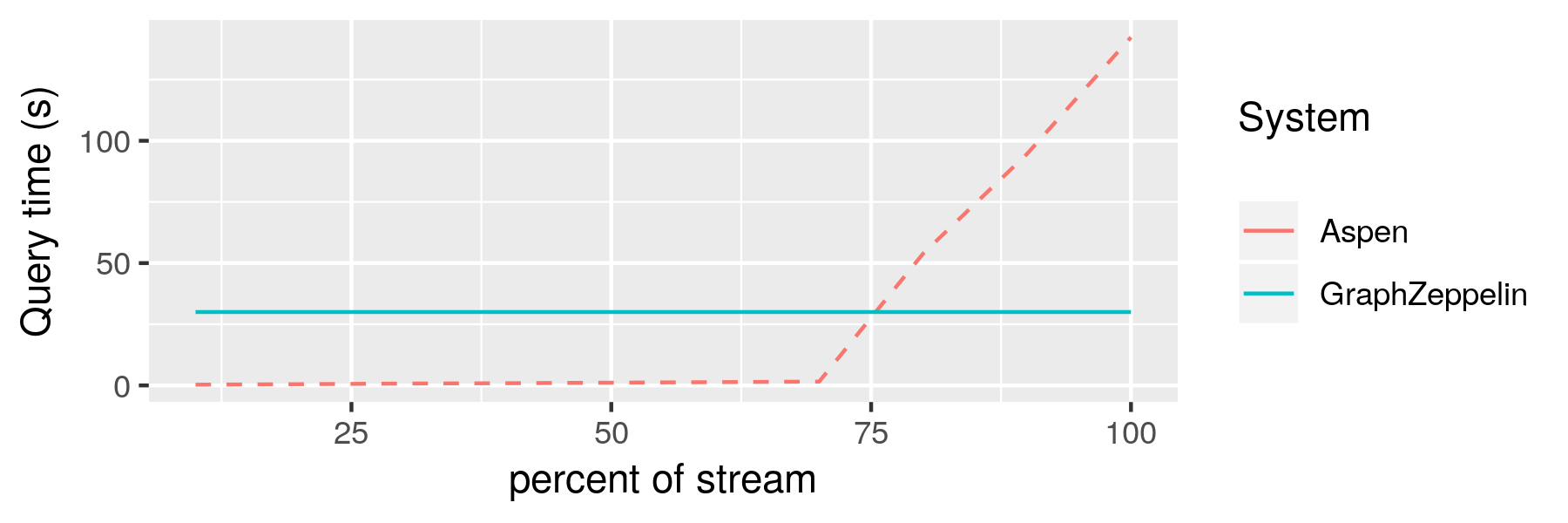}

\setlength{\belowcaptionskip}{-16pt} 
    \caption{
On-disk query times.
    }\label{fig:query2}
  \end{subfigure}
\setlength{\belowcaptionskip}{-16pt} 
  \caption{\sysname query performance is comparable to or better than Aspen and Terrace for dense graphs.}
  \label{fig:query}
\end{figure}
\newtext{
We show experimentally that \sysname gives comparable query performance to Aspen and Terrace on dense graphs when all systems' data structures fit in RAM. When their data structures reside on disk, \sysname answers queries more than five times faster than Aspen (and Terrace ingests too slowly to test).
}

\newtext{
\sysname's buffering strategies create a tradeoff between stream-ingestion rate and query latency.
When \sysname receives a connectivity query, it must process remaining stream updates in its buffering system before computing connectivity using Boruvka's algorithm. 
Large buffers improve stream-ingestion rate (see Section \ref{sec:buffer-size-exp}), particularly when sketches are stored on disk, but this comes at the cost of increased query latency since these large buffers must be emptied. For the same reasons, small buffers improve query latency but may decrease ingestion rate.
}

\newtext{
Figure \ref{fig:query1} compares the query latency of \sysname, Aspen, and Terrace on the kron17 stream where connectivity queries are issued every 10\% of the way through the stream. In this experiment \sysname used small 400-byte leaf-only buffers, enough space for 100 stream updates. At the beginning of the stream, when the graph is sparser, both Aspen and Terrace answer queries more quickly than \sysname. As the stream progresses and the graph becomes denser, \sysname's query time stays constant while Aspen's and Terrace's increase. By 70\% of the way through the stream \sysname is fastest.  Even with \sysname's small buffer size its ingestion rate was 3.95 million per second, twice as fast as Aspen and almost 100 times faster than Terrace.
}

\newtext{
Figure \ref{fig:query2} compares the query latency of \sysname and Aspen when RAM is limited to 12GiB, forcing both systems to store part of their data structures on disk. Terrace ingests too slowly given only 12GiB of RAM to be included in the experiment. In this experiment \sysname used 8.3 KB leaf-only buffers (one-tenth of sketch size). \sysname takes 24 seconds to perform queries regardless of graph density. Aspen's queries are fast until the graph is too dense to fit in RAM; its last query takes 142 seconds, five times slower than \sysname.  Notably, \sysname maintains an ingestion rate of 4.15 million updates per second, 46 times faster than Aspen. Both systems spend the majority of time on insertions, where \sysname's advantages come through. 
}

\vspace{-.5em}
\section{Related Work}


\paragraph{Graph Streaming Systems.}
Existing graph stream processing systems are designed primarily to handle updates in batches consisting entirely of insertions or entirely of deletions. Streaming systems that process updates in batches are generally divided into two categories. The first (which includes Terrace) consists of those systems which finish ingestion prior to beginning queries and finish queries prior to accepting any additional edges~\cite{terrace,ammar2018distributed,busato2018hornet,ediger2012stinger,murray2016incremental,sengupta2016graphin,sengupta2017evograph}. The second (which includes Aspen) allows updates to be applied asynchronously by periodically taking ``snapshots'' of the graph during ingestion to be used in conducting queries~\cite{aspen,cheng2012kineograph,iyer2015celliq,iyer2016time,macko2015llama}.

The batching employed in these systems amortizes the cost of applying updates, but also limits the granularity at which insertions and deletions may be interspersed during ingestion. In contrast, \sysname allows for insertions and deletions to be arbitrarily interspersed during ingestion without sacrificing query correctness. 

\paragraph{External Memory Systems.}
There is a rich literature of graph processing systems process static graphs in external memory. Some such systems store the entire graph out-of-core~\cite{KyrolaBlGu12,han2013turbograph,zhu2015gridgraph,maass2017mosaic,zhang2018wonderland}, and others are semi-external memory systems that maintain only the vertex-set in RAM~\cite{roy2013x,zheng2015flashgraph,yuan2016pathgraph,liu2017graphene,ai2018clip}. Some systems provide (at least theoretical) design extensions to handle queries on graphs with insert-only updates~\cite{KyrolaBlGu12,zhang2018wonderland,vora2016load,cheng2015venus,vora2019lumos}, but to the best of our knowledge \sysname is the first to leverage external-memory effectively in the streaming model of insertions \textit{and} deletions.


\section{Conclusion}
\sysname computes the connected components of graph streams using space asymptotically smaller than an explicit representation of the graph. It is based on \sketchname, a new $\ell_0$-sketching data structure that outperforms the state of the art on graph-streaming workloads.  
This new sketching technique allows \sysname to process larger, denser graphs than existing graph-streaming systems given a fixed RAM budget and to ingest these graph streams more quickly. 
Even when \sysname's sketch data structures are too large to fit in RAM, its work-buffering strategies allow it to process graph streams on SSD at the cost of a only small decrease in ingestion rate.
Thus, \sysname is simultaneously a space-optimal graph semi-streaming algorithm and an I/O-efficient external-memory algorithm.


The small space complexity of \sysname's linear sketch is optimized for large, dense graphs, unlike prior graph-processing systems, which often focus on sparse graphs. 
Thus, \sysname demonstrates that computational questions on graphs once thought intractably large and dense are now within reach. 

\newtext{
Currently large, dense graphs are studied rarely and at great cost on large high-performance clusters~\cite{facebookdense}. Since \sysname's sketches can be updated independently (Section~\ref{subsec:multithreading}), we believe that they can be partitioned throughout a distributed cluster without sacrificing stream ingestion rate. 
}

\sysname illustrates that additional algorithmic improvements  help make graph semi-steaming algorithms into a powerful engineering tool
by reducing the update-time complexity and allowing sketches to be stored efficiently on SSD.  These techniques may generalize to other graph-analytics problems.  




\subsection*{Acknowledgments}

\bibliographystyle{abbrv}
\bibliography{all}
\clearpage
\appendix
\end{document}